\documentclass[superscriptaddress,groupedaddress,aps,pra,twocolumn,showpacs,nofootinbib]{revtex4}
\usepackage{etex}
\usepackage[small,nohug,heads=littlevee]{diagrams}
\diagramstyle[labelstyle=\scriptstyle]

\usepackage{amsmath,amssymb,amsthm}
\usepackage{easybmat}
\usepackage{hyperref}
\usepackage[pdftex]{graphicx}
\usepackage{times,txfonts}

\newtheorem{thm}{Theorem}[section]
\newtheorem{lem}[thm]{Lemma}

\newtheorem{cor}[thm]{Corollary}

\newtheorem{defn}{Definition}[section]

\theoremstyle{definition}
\newtheorem{exmp}{Example}[section]
\theoremstyle{remark}

\newcommand{\beq}{\begin{equation}}
\newcommand{\eeq}{\end{equation}}

\newcommand{\tinyspace}{\mspace{1mu}}
\newcommand{\microspace}{\mspace{0.5mu}}

\newcommand{\tr}{\operatorname{Tr}}

\def\complex{\mathbb{C}}
\def\real{\mathbb{R}}

\def\I{\mathbb{I}}
\def \lket {\left|}
\def \rket {\right\rangle}
\def \lbra {\left\langle}
\def \rbra {\right|}
\newcommand{\ket}[1]{\lket\microspace #1 \microspace\rket}
\newcommand{\bra}[1]{\lbra\microspace #1 \microspace\rbra}
\newcommand{\proj}[1]{\ket{#1}\bra{#1}}
\newcommand{\comp}[1]{#1^*}
\newcommand{\state}[2]{| #1\rangle\langle #2|}
\newcommand{\Tr}{\operatorname{Tr}}

\newcommand{\diag}[1]{\operatorname{diag}(#1)}
\renewcommand{\neg}{\mathcal{N}}
\newcommand{\lnorm}[1]{\left\lVert\tinyspace #1 \tinyspace\right\rVert_{l_1}}
\newcommand{\upto}[1]{\{1,2,\cdots, #1\}}
\newcommand{\CC}{\operatorname{\mathcal{CC}}}
\newcommand{\MCS}{\operatorname{\mathcal{MCS}}}
\newcommand{\CQ}{\operatorname{\mathcal{CQ}}}

\newcommand{\matsize}[1]{\mathbb{M}(#1)}
\newcommand{\bas}[1]{\mathcal{B}(#1)}
\newcommand{\ppt}{\operatorname{\mathcal{PPT}}}
\newcommand{\sep}{\operatorname{\mathcal{SEP}}}
\newcommand{\rel}[2]{\operatorname{S}(#1||#2)}

\makeatletter
\let\orgdescriptionlabel\descriptionlabel
\renewcommand*{\descriptionlabel}[1]{%
  \let\orglabel\label
  \let\label\@gobble
  \phantomsection
  \edef\@currentlabel{#1}%
  \let\label\orglabel
  \orgdescriptionlabel{#1}%
}
\makeatother
\makeatletter
\newcommand{\rmnum}[1]{\romannumeral #1}
\newcommand{\Rmnum}[1]{\expandafter\@slowromancap\romannumeral #1@}
\makeatother

\begin{document}
\title{Negativity of quantumness and its interpretations}
\author{Takafumi Nakano}
\affiliation{Institute for Quantum Computing and Department of Physics and Astronomy, University of Waterloo, Waterloo ON
N2L 3G1, Canada}

\author{Marco Piani}
\affiliation{Institute for Quantum Computing and Department of Physics and Astronomy, University of Waterloo, Waterloo ON
N2L 3G1, Canada}

\author{Gerardo Adesso}
\affiliation{School of Mathematical Sciences, University of Nottingham, University Park, Nottingham NG7 2RD, United Kingdom}

\begin{abstract}
We analyze the general nonclassicality of correlations of a composite quantum systems as measured by the negativity of quantumness. The latter corresponds to the minimum entanglement, as quantified by the negativity, that is created between the system and an apparatus that is performing local measurements on a selection of subsystems. The negativity of quantumness thus quantifies the degree of nonclassicality on the measured subsystems. We demonstrate a number of possible different interpretations for this measure, and for the concept of quantumness of correlations in general. In particular, for general bipartite states in which the measured subsystem is a qubit, the negativity of quantumness acquires a geometric interpretation as the minimum trace distance from the set of classically correlated states. This can be further reinterpreted as minimum disturbance, with respect to trace norm, due to a local measurement or a nontrivial local unitary operation. We calculate the negativity of quantumness in closed form for Werner and isotropic states, and for all two-qubit states for which the reduced state of the system that is locally measured is maximally mixed---this includes all Bell diagonal states. We discuss the operational significance and potential role of the negativity of quantumness in quantum information processing.
\end{abstract}


\maketitle

\section{Introduction}\label{sec:introduction}

Quantum systems differ from classical ones in a number of ways. This is particularly true for composite systems, which can exhibit quantum features like nonlocality and quantum entanglement~\cite{horodecki2009quantum}. They are certainly striking manifestations of a deviation from classicality. While quantum entanglement can be regarded as one of the most characteristic traits of quantum mechanics~\cite{CambridgeJournals:1737068},  a lot of effort has recently been directed towards characterizing a more general notion of \emph{quantumness} of correlations~\cite{ollivier2001quantum,henderson2001classical,modi:2011a}, almost always present in mixed states also in the absence of entanglement~\cite{acinferraro}. The study of quantum correlations, or of quantumness in its most essential manifestation in composite systems, has a deep foundational value as it provides a trail to investigate the boundary between the classical and the quantum world from the perspective of quantum measurements \cite{zurek2003decoherence}. Quantumness of correlations manifests, for instance, when any local complete projective measurement on a subsystem necessarily alters the state of a composite system \cite{ollivier2001quantum}. There are many other (sometimes equivalent) signatures to reveal quantumness of correlations in a state, and there are correspondingly a number of possible ways to quantify such quantumness (including the {\it quantum discord} \cite{ollivier2001quantum,henderson2001classical}), which range from informational to geometric and thermodynamical settings. Some of the most prominent approaches are summarized below, while for a more extensive treatment we defer the reader to a recent review, see \cite{modi:2011a} and references therein.

From an applicative point of view, Quantum Information Processing aims at harnessing quantum properties  to outperform classical information processing \cite{nichu}. A natural step towards such a goal is that of developing concepts and tools to more precisely determine which states possess or do not possess a certain quantum property, further aiming to quantifying and exploiting the latter when present. While quantumness of correlations reduces to entanglement for pure states, thus already embodying the key resource for quantum information processing in absence of noise \cite{horodecki2009quantum,nichu}, a number of researchers are investigating the role of general nonclassical correlations in quantum computation \cite{datta,dqc1expwhite,dqc1explaflamme,dattacomputing}, quantum communication~\cite{cavalcanti:pr2011a,np2} and quantum metrology \cite{modix,LQU}. 


In the following, we will always address the notions of classicality or quantumness as referred to the correlations among subsystems in the state of a composite system; we similarly adopt the wording of (non)classical states to mean equivalently (non)classically correlated states. We will not be concerned with other definitions of (non)classicality such as those usually adopted in quantum optics to characterize the nature of light \cite{nonclasslight}, which often lead to a very different classification of states into classical and non-classical~\cite{PhysRevLett.108.260403}.

As anticipated, there are a wide variety of approaches to quantify the nonclassicality (of correlations) of the state $\rho$ of a quantum system \cite{modi:2011a}. To list a few, it can be measured by
\begin{description}
\item[Approach 1\label{itm:activation}]
(Activation) the minimum amount of entanglement created between the system and its measurement apparatus in a local measurement~\cite{streltsov:prl2011b,piani:prl2011a,doi:10.1142/S0219749911008258,piani:pr2012a};
\item[Approach 2\label{itm:distance}]
(Geometric) the minimum distance between $\rho$ and its closest classical state~\cite{modi:prl2010a,dakic:prl2010a,luo2010geometric,problemgeometric};
\item[Approach 3\label{itm:decoherence}]
(Disturbance by measurement) the minimum amount of disturbance caused by local projective measurements~\cite{ollivier2001quantum,henderson2001classical,luo:pr2008a};
\item[Approach 4\label{itm:localunitary}]
(Disturbance by unitary) the minimum disturbance caused by particular local unitaries~\cite{gharibian:2012a,faberdagmar}.
\end{description}

As summarized in \cite{modi:2011a}, a number of nonclassicality measures can be defined via \ref{itm:distance}, \ref{itm:decoherence} and \ref{itm:localunitary} using different distance functions such as the Hilbert-Schmidt distance or the quantum relative entropy.
Even though at a first glance \ref{itm:activation} seems to be very different from the other approaches, it turns out it is intimately connected with them (see also \cite{coles:2012a}). This is because there are entanglement measures that are similarly defined as the distance from the closest separable state. One such example is given by the relative entropy of entanglement~\cite{vedral:prl1997a}, which is a well known upper bound to the distillable entanglement~\cite{rains:itit2001a}.

In this paper we develop an extensive study of a promising measure of nonclassical correlations recently defined in~\cite{piani:prl2011a} along \ref{itm:activation}: the \emph{Negativity of Quantumness} (NoQ).
NoQ~\cite{piani:prl2011a} corresponds to the adoption of negativity~\cite{negativity} as entanglement measure in  \ref{itm:activation}. In~\cite{piani:pr2012a}  it was proven that   \ref{itm:activation} leads to a quantitative hierarchy of correlations that formalizes the intuition that `the general quantumness of correlations is more than entanglement'. In particular, the entanglement generated in a complete projective local measurement is provably always greater than the entanglement present in the state whose non-classicality of correlations is under scrutiny. This fact is independent of the specific choice of entanglement measure. Nonetheless, given the usefulness of negativity~\cite{negativity} in the study and quantification of entanglement---a usefulness that comes in particular by its being easy to calculate---it is natural to focus on the specific hierarchy it generates. Furthermore, in this article we prove there are other reasons to focus on NoQ.


Indeed, we find that NoQ can be further interpreted from the perspective of~\ref{itm:distance}, \ref{itm:decoherence}, and~\ref{itm:localunitary} under suitable conditions. In particular, in the special case of a bipartite state $\rho_{AB}$, when furthermore the measured subsystem $A$ is a two-level quantum object (qubit), it turns out that the NoQ acquires a geometric interpretation as the minimum distance between $\rho_{AB}$ and the set of classical states, with the distance measured in {\it trace norm}. This can be proven to be further equivalent to the minimal state change, again in trace norm, after either a local measurement on $A$ or a (nontrivial) local unitary evolution on $A$. The equivalence between the last two approaches when $A$ is a qubit is here proven for any norm-based distance, complementing the original results of \cite{gharibian:2012a,faberdagmar} which were specific to the Hilbert-Schmidt norm. All these results are derived and described throughout the paper with relevant examples. Our work also provides a proof to a conjecture raised by Khasin {\it et al.}~about a bound for the negativity (of entanglement) \cite{khasin:pr2007a}. We obtain closed analytical expressions for the NoQ in relevant cases such as Werner and isotropic states of arbitrary dimension \cite{werner,iso}, and a family of two-qubit states where qubit $A$ is maximally mixed, which includes Bell diagonal states. In the latter instance, the problem is recast into an appealing geometrical formulation. The closed formula for two-qubit states with one maximally mixed marginal, together with the hierarchical relation between quantumness of correlations and entanglement of Ref.~\cite{piani:pr2012a} allows, e.g., a consistent study and comparison of the evolution of entanglement---as measured by negativity---and quantumness of correlations---as measured by the one-sided negativity of quantumness---under the action of a family of qubit channels, e.g., a semigroup.

We note that, during completion of this manuscript, we became aware of other works investigating a distance-based measure of quantumness based on trace norm  \cite{vianna,giovannettiprep,faberprep,saro}

The paper is organized as follows. In Section~\ref{sec:notation} we fix some notation adopted throughout the manuscript. In Section~\ref{sec:CNC} we review the main definitions and approaches to quantify quantumness of correlations as sketched above. Section~\ref{sec:NoQ} is focused on the definition and formulation of the NoQ. The main properties and interpretations of the NoQ are discussed in Section~\ref{sec:interp}, along with its interplay with the usual negativity of entanglement.  In Section~\ref{sec: qubit system} we prove in general the equivalence of the various approaches when the measured subsystem is a qubit. In Section~\ref{sec:examples} we calculate the NoQ for relevant families of bipartite states. We conclude the main body of the paper in Section~\ref{sec:concl}. The Appendices contain a number of technical proofs and extensions.


\section{Notation}
\label{sec:notation}

We will deal only with finite-dimensional Hilbert spaces and we will identify linear operators with matrices. We will denote by $\langle A, B \rangle:= \Tr(A^\dagger B)$ the \emph{Hilbert-Schmidt inner product} of two matrices $A$ and $B$.

The \emph{Schatten $p$-norms} $\|\cdot\|_p$ are a family of norms parametrized by a real number $p\geq 1$. If $\sigma_i(A)$'s are the singular values of a matrix $A$, the $p$-norm of the latter is defined as \begin{equation}\|A\|_p:=\left(\sum_i \sigma^p_i(A)\right)^{1/p}.\end{equation}
In particular, we are going to focus our attention on the $1$-norm (also called \emph{trace norm}) $\|A\|_1=\sum_i \sigma_i(A)=\Tr\sqrt{A^\dagger A}$,
 and on the $2$-norm $\|A\|_2=\sqrt{\sum_i\sigma_i(A)^2}=\sqrt{\Tr(A^\dagger A)}=\sqrt{\langle A, A \rangle}$ (also called Hilbert-Schmidt norm). 
For any norm $\|\cdot\|$, one can define an associated distance on matrices by means of $\|A-B\|$. In particular, the distance associated with the $1$-norm is also called \emph{trace distance} (up to a factor $1/2$), while the distance associated with the $2$-norm is known as \emph{Hilbert-Schmidt distance}. The trace distance between two mixed states (i.e., positive semidefinite operators of trace one) has a direct operational interpretation linked to the probability of success in distinguishing the two states via a measurement~\cite{nichu}.

We will also make use of the $l_1$-norm, which is a basis-dependent norm defined as the sum of the absolute values of the entries of a matrix: $\|A\|_{l_1}=\sum_{i,j}|A_{i,j}|$. See Appendix~\ref{sec: l1-norm} for more details.

The relative entropy of a density matrix $\rho$ with respect to a density matrix $\sigma$ is defined as $S(\rho\|\sigma):=\Tr(\rho(\log \rho -\log \sigma))=-S(\rho)-\Tr(\rho\log\sigma)$. Here $S(\rho):=-\Tr(\rho\log\rho)$ is the von Neumann entropy, and the logarithms are taken in base $2$ throughout all the paper. The relative entropy is not a distance as, for example, it is not symmetric in its arguments, but for the sake of our investigation we can and will treat it as if it was a distance measure. This is done routinely in quantum information theory \cite{vedralrmp}.

A \emph{channel} is  a completely positive trace-preserving (CPTP) linear map on operators~\cite{nichu}; it admits a Kraus representation
\begin{equation}
\label{eq:channel}
\Lambda[B] = \sum_i K_i B K_i^\dagger,
\end{equation}
with $K_i$'s the corresponding Kraus operators satisfying the trace-preservation condition $\sum_i K_i^\dagger K_i =\I$. The dual $\Lambda^\dagger$ of a channel $\Lambda$---and in general, of a linear map---is defined via the Hilbert-Schmidt inner product through the relation
\[
\langle A, \Lambda [B] \rangle = \langle \Lambda^\dagger [A], B\rangle, \quad \forall A, B.
\]
It easy to verify that for a channel with Kraus decomposition \eqref{eq:channel} the dual map has Kraus decomposition
\[
\Lambda^\dagger[A] = \sum_i K_i^\dagger A K_i.
\]
Hence, $\Lambda$ being trace-preserving implies that $\Lambda^\dagger$ is unital, i.e. $\Lambda^\dagger[\I]=\I$; $\Lambda^\dagger$ is furthermore a channel---i.e., also trace-preserving---if $\Lambda$ is unital itself.

We will also make use of the Pauli matrices $\sigma_1=\sigma_x=\begin{pmatrix} 0 & 1 \\ 1 & 0 \end{pmatrix}$, $\sigma_2=\sigma_y=\begin{pmatrix} 0 & -i \\ i & 0 \end{pmatrix}$, and $\sigma_3=\sigma_z=\begin{pmatrix} 1 & 0 \\ 0 & -1 \end{pmatrix}$, which, alongside with the identity $\sigma_0=\I=\begin{pmatrix} 1 & 0 \\ 0 & 1 \end{pmatrix}$, form a basis for the space of $2\times 2$ matrices.


\section{Classicality and nonclassicality of correlations: notions and measures}\label{sec:CNC}

\subsection{Classicality of quantum states}

One key quantum feature is the fact that a measurement will in general disturb the state of the system being measured. Most importantly, a local measurement will typically lead to a decrease in (total) correlations -- quantified, e.g.,  in terms of quantum mutual information -- between the measured system and any other systems that might have been initially correlated with it~\cite{ollivier2001quantum,henderson2001classical,piani:prl2008a}. Because of the orthogonality and perfect distinguishability of the states that form an orthonormal basis, a complete projective measurement can be thought as a quantum-to-classical mapping, where the information extracted from the quantum system is recorded in a classical register. In this sense, only correlations between a classical system and the remaining unmeasured systems are left after the local measurement. It has been actually suggested that such surviving correlations should be deemed \emph{classical}~\cite{henderson2001classical}. It can be proven that a local measurement that does not destroy any amount of correlations exists  for a state if and only if the measured system could be considered classical to start with~\cite{ollivier2001quantum,luo:pr2008a,piani:prl2008a}, so that such a measurement does not disturb the system. To be more precise, the following notions will be adopted.

\begin{defn}[Classicality of a quantum state (\romannumeral 1)]
Let $\rho_{\upto{n}}$ be an $n$-partite quantum state. For any $i\in\upto{n}$, $\rho_{\upto{n}}$ is \emph{classical on the $i$-th system} if there exists a complete local projective measurement on the $i$-th subsystem which leaves $\rho_{\upto{n}}$ invariant. 
\end{defn}

Complete projective measurements are described by a set of orthogonal rank-one projectors, say $\{\state{a_i}{a_i}\}$, such that they sum up to the identity of the space, i.e., $\sum_i\state{a_i}{a_i}=\I$. Therefore the state is invariant under such measurement if and only if the original state has block-diagonal form with respect to the basis $\{\ket{a_i}\}$, and this can be used as an alternative definition of classicality of the state.

\begin{defn}[Classicality of a quantum state (\romannumeral 2)]
Let $\rho_{\upto{n}}$ be an $n$-partite quantum state. For any $i\in\upto{n}$, $\rho_{\upto{n}}$ is \emph{classical on the $i$-th system} if $\rho_{\upto{n}}$ can be represented as
\beq
\label{eq:classical}
\rho_{\upto{n}}=\sum_j\state{a_j}{a_j}_i\otimes \sigma^j_{\{1,\cdots,i-1,i+1,\cdots,n\}},
\eeq
where $\left\{\ket{a_j}\right\}$ is an orthonormal basis of the $i$-th system and
\[
\sigma^j_{\{1,\cdots,i-1,i+1,\cdots,n\}}=\langle a_j|\rho_{\upto{n}}|a_j\rangle.
\]
\end{defn}

In this paper, mainly bipartite quantum states are considered, $\rho\equiv \rho_{AB}$. A bipartite quantum state which is classical only on one subsystem (say $A$) is called \emph{classical-quantum} ($\CQ$). Such a state will exhibit zero quantumness (of correlations between $A$ and $B$) with respect to local measurements on $A$. Similarly, a state classical on both subsystems is called \emph{classical-classical} ($\CC$) \cite{piani:prl2008a}.

Classicality of subsystems is a much stronger notion than separability, i.e., absence of entanglement, as recalled below.

\begin{defn}[Separable and entangled states] A state $\rho_{AB}$ is \emph{separable} in the bipartition $A$ versus $B$, also indicated as the $A:B$ bipartition,  if it can be represented as \cite{horodecki2009quantum}
\beq
\label{eq:entangled}
\rho_{AB}=\sum_j p_j \tau_A^j\otimes \sigma_B^j,
\eeq
with $(p_j)_j$ a probability distribution and $\tau_A^j$, $\sigma_B^j$ quantum states for $A$ and $B$, respectively. A state $\rho_{AB}$ is $A:B$ \emph{entagled} if it is not $A:B$ separable.
\end{defn}
It is clear from \eqref{eq:classical} and \eqref{eq:entangled} that every state $\rho_{AB}$ that is classical on $A$ is $A:B$ separable, while the opposite does not hold.

Given a state which is not classical on some subsystem $i$, i.e., such that the state would get disturbed by any possible complete projective measurement on subsystem $i$, it is natural to try to quantify the amount of non-classicality of correlations in that state, in particular from an operational perspective.
As mentioned in~Sec.~\ref{sec:introduction}, there are different approaches to quantify nonclassicality that we briefly review.


\subsection{Non-classicality by system-apparatus entanglement}

We now recall more in detail  how a measure---or rather, a family of measures---of nonclassical correlations can be based on~\ref{itm:activation}. The definition goes through the consideration of a particular set of states produced during the measurement of a (sub)system, when such process is modeled by a particular unitary interaction---which we call \emph{measurement interaction}---between the system and a measurement apparatus. We call the states that are the result of such an interaction \emph{pre-measurement states}~\cite{zurek2003decoherence}.

\begin{defn}[Measurement interaction]
A \emph{measurement interaction} $V_{A\mapsto AA'}$ on system $A$ is a linear isometry from $A$ to a bipartite system $AA'$ where $A'$ has the same dimension as $A$.  This isometry is defined by the following mapping of an orthonormal basis of A, $\{\ket{a_k}\}$:
\[
V_{A\mapsto AA'}\ket{a_k}_A=\ket{a_k}_A\ket{k}_{A'},\quad \forall k,
\]
where $\{\ket{k}_{A'}\}$ is the computational basis of system $A'$. This operation is defined in a basis-dependent way, i.e., the choice of different local orthonormal basis of A results in distinct measurement interactions. Therefore the operation is sometimes denoted as $V_{A\mapsto AA'}^{\{\ket{a_k}\}}$ when the basis needs to be specified.
\end{defn}

\begin{defn}[Pre-measurement state]
Let $\rho_{\upto{n}}$ be an $n$-partite quantum state. For some choice of subsystems $\Sigma\subseteq \upto{n}$ and of measurement bases  $\{\ket{a_k}\}$ (one for each subsystem $i\in\Sigma$), the corresponding pre-measurement state for $\rho$ is
\[
\tilde{\rho}_{\Xi}:=\left(\bigotimes_{i\in\Sigma}\mint{i}\right)\rho\left(\bigotimes_{i\in\Sigma}\mint{i}\right)^\dagger,
\]
 where $\Xi=\upto{n}\cup\Sigma'$, i.e, $\tilde{\rho}_{\Xi}$ is the joint state of the initial quantum systems and their measurement apparatuses.
\end{defn}

As already mentioned, measurement interactions depend on the choice of a specific basis for each subsystem in $\Sigma$. The action of measurement interactions defined for distinct bases results---in general---in distinct final pre-measurement states. This leads to the consideration of the entire class of potential pre-measurement states.

The next theorem asserts that the system-apparatus separability of the pre-measurement states characterizes the classicality of the correlations of the original system.

\begin{thm}[Activation protocol~\cite{piani:prl2011a,streltsov:prl2011b, doi:10.1142/S0219749911008258, piani:pr2012a}]\label{thm: activation protocol}
A $n$-partite quantum state $\rho_{\upto{n}}$ is classical on its subsystems $\Sigma\subseteq \upto{n}$ if and only if there exists a corresponding ${\upto{n}}:\Sigma'$ separable pre-measurement state $\tilde{\rho}_{\upto{n}\cup\Sigma'}$.
\end{thm}

One can introduce a family of quantifiers of the nonclassical correlations present in a quantum state exploting this intimate connection between the nonclassicality of a quantum state and the entanglement properties of the corresponding pre-measurement states.

\begin{defn}[Non-classicality by system-apparatus entanglement~\cite{piani:prl2011a,streltsov:prl2011b, doi:10.1142/S0219749911008258, piani:pr2012a}]\label{def:quanumnessmeasure}
Let $\rho_{\upto{n}}$ be a $n$-partite quantum state and $\Sigma\subseteq \upto{n}$ be a set of its subsystems. Then the measure of nonclassicality of correlation (or quantumness, in short) present in ${\rho}_{\upto{n}}$ as revealed on subsystems $\Sigma$ with respect to an entanglement measure $E$ is defined as
\begin{equation}
\label{eq:defquantent}
Q_E^\Sigma(\rho_{\upto{n}}):= \min_{\bigotimes_{i\in\Sigma}\bas{i}} E_{\upto{n}:\Sigma'}(\tilde{\rho}_{\upto{n}\cup\Sigma'}),
\end{equation}
where $\bas{i}$ denotes a local orthonormal basis of subsystem $i\in\Sigma$, the bipartite entanglement is measured across the bipartite cut $\upto{n}:\Sigma'$, and the minimum is taken over different pre-measurement states for different measurement interaction defined by $\bigotimes_{i\in\Sigma}\bas{i}$.
\end{defn}



\subsection{Other measures of nonclassicality}
In this susbsection, we will briefly summarize other approaches to define measures of nonclassicality of quantum states, namely \ref{itm:distance} and \ref{itm:decoherence}. The approach based on \ref{itm:localunitary} and its relevance to our other results will be discussed in~Sec.~\ref{sec: qubit system}.

\subsubsection{Measures of nonclassicality based on disturbance}
Since the nonclassicality of a quantum state is defined by the unavoidable disturbance caused by any local projective measurement, perhaps the most immediate way to study how nonclassical a quantum state is, is to quantify the difference between the state before and after a measurement (\ref{itm:decoherence}).

\begin{defn}[Non-classicality by measurement disturbance~\cite{luo:pr2008a,horodecki2005local}]\label{def:measurementdisturbance}
Let $d(\cdot,\cdot)$ be a distance function. The nonclassicality of a $n$-partite quantum state $\rho_{\upto{n}}$ revealed on its subsystems $\Sigma\subseteq\upto{n}$, measured by the minimum disturbance caused by local projective measurements on each subsystem in $\Sigma$, with respect to the distance function $d(\cdot,\cdot)$, is defined as
\begin{equation}
\mathcal{D}_{d(\cdot,\cdot)}^\Sigma(\rho_{\upto{n}}):=
\min_{\bigotimes_{i\in\Sigma}\bas{i}}d\left(\rho_{\upto{n}},\bigotimes_{i\in\Sigma}\Pi_{\bas{i}}[\rho_{\upto{n}}]\right),
\end{equation}
where $\Pi_{\bas{i}}$ denotes a complete projective measurement on subsystem $i$ on a complete orthonormal basis $\bas{i}$ and the minimum is taken over different choices of local bases for these projective measurements. If the distance function $d(\cdot,\cdot)$ is derived from a norm $\|\cdot\|$, we will use the shorthand notation $\mathcal{D}_{\|\cdot\|}^\Sigma(\rho_{\upto{n}})$.
\end{defn}
\begin{exmp}
We list a few measures of quantumness of correlations for bipartite states $\rho_{AB}$ based on this notion:
\begin{itemize}
	\item \emph{Zero-way deficit} \cite{horodecki2005local}:
	\begin{equation}
	\label{eq:zeroway}
	\begin{split}
	\Delta^\emptyset(\rho_{AB}):&= \mathcal{D}_{\rel{\cdot}{\cdot}}^{AB}(\rho_{AB})\\
	&=\min_{\bas{A}\otimes\bas{B}} S(\rho_{AB}\|(\Pi_{\bas{A}}\otimes\Pi_{\bas{B}})[\rho_{AB}]).
	\end{split}
	\end{equation}

 	The distance function here is the quantum relative entropy (see Section \ref{sec:notation}).
 	\item \emph{One-way deficit} \cite{horodecki2005local}:
	\begin{align*}
	\Delta^\rightarrow(\rho_{AB}):&=\mathcal{D}_{\rel{\cdot}{\cdot}}^{A}(\rho_{AB})\\
	&=\min_{\bas{A}} S(\rho_{AB}\|\Pi_{\bas{A}}[\rho_{AB}]).
	\end{align*}
	Again the distance function is the quantum relative entropy. Notice that the one-way deficit vanishes on $\CQ$ states while the zero-way deficit is a symmetric measure vanishing only on $\CC$ states.
	\item \emph{Geometric discord} \label{itm:geometricdiscord}\cite{dakic:prl2010a,luo2010geometric}:
	\begin{equation}
	\label{eq:geometricdiscord}
	\begin{split}
	 D_G^A(\rho_{AB}):&=\mathcal{D}_{\|\cdot\|_2^2}^{A}(\rho_{AB})\\
	&=\min_{\bas{A}}||\rho_{AB}-\Pi_{\bas{A}}[\rho_{AB}]||_2^2.
	\end{split}
	\end{equation}
 	The distance function here is the square of the Hilbert-Schmidt distance (see Section \ref{sec:notation}). For a discussion of some conceptual issues with the use of the geometric discord as quantumness measure, see~\cite{problemgeometric}.
\end{itemize}
\end{exmp}

\subsubsection{Distance-based measures of the nonclassicality of correlations}
The quantification of nonclassicality based on \ref{itm:distance}  follows a very common approach in quantum information theory. Here the quantumness of correlations  is defined as the minimum distance from the (suitably chosen) set of classical states. 

\begin{defn}[Non-classicality by distance from classical states]
Let $d(\cdot,\cdot)$ be any distance function. The nonclassicality of a $n$-partite quantum state $\rho_{\upto{n}}$ on its subsystems $\Sigma\subseteq\upto{n}$, measured by the distance from the set $\mathcal{C}(\Sigma)$ of states which are classical on each subsystem in $\Sigma$, with respect to the distance function $d(\cdot,\cdot)$, is defined as
\begin{equation}
\mathbb{D}_{d(\cdot,\cdot)}^\Sigma(\rho_{\upto{n}}):=
\min_{\eta\in\mathcal{C}(\Sigma)}d(\rho_{\upto{n}},\eta).
\end{equation}
 If the distance function $d(\cdot,\cdot)$ is derived from a norm $\|\cdot\|$, we will use the shorthand notation $\mathbb{D}_{\|\cdot\|}^\Sigma(\rho_{\upto{n}})$.
\end{defn}
\begin{exmp} Here we provide two examples of nonclassicality measures for a bipartite state $\rho_{AB}$ based on this notion that differ in  the choice of distance measures and relevant set of classical states (recall that $\CC$ stands for the set of states which are classical both on system $A$ and $B$ and $\CQ$ is the set of states which are classical on $A$):
\begin{itemize}
	\item \emph{Relative entropy of discord}~\cite{modi:prl2010a}:
	\begin{align*}
		D_R(\rho_{AB}):&=\mathbb{D}_{\rel{\cdot}{\cdot}}^{AB}(\rho_{AB})\\
		&=\min_{\eta\in\CC} \rel{\rho_{AB}}{\eta}.
	\end{align*}
	Here the distance function is the quantum relative entropy. It can be proven that the relative entropy of discord is equivalent to the zero-way deficit~\eqref{eq:zeroway}~\cite{modi:prl2010a}.
	\item \emph{Geometric discord}~\cite{dakic:prl2010a}:
	\begin{equation}
	\label{eq:geometricdiscord2}
	\begin{split}
		D_G(\rho_{AB}):&=\mathbb{D}_{\|\cdot\|_2^2}^{A}(\rho_{AB})\\
		&=\min_{\eta\in\CQ}||\rho_{AB}-\eta||_2^2.
	\end{split}
	\end{equation}
	Here the distance function is (the square of) the Hilbert-Schmidt distance. It is easily verified that the two definitions~\eqref{eq:geometricdiscord} and \eqref{eq:geometricdiscord2} for the geometric discord are equivalent \cite{luo2010geometric}.
\end{itemize}  
\end{exmp}


\section{Negativity of Quantumness}\label{sec:NoQ}

The activation protocol (Definition \ref{def:quanumnessmeasure}) allows us to define a measure of nonclassical correlations $Q_E$ for each entanglement measure $E$ we may want to consider. Through this mapping, some entanglement measures generate known nonclassicality measures and others generates new ones. Negativity \cite{negativity} is a widely used entanglement measure with a very appealing property: it is easily computable\footnote{Unlike other operationally more meaningful entanglement measures such as entanglement cost, distillable entanglement and entanglement of formation, which typically involve a nontrivial optimization over a large set of parameters~\cite{horodecki2009quantum}.}. For the rest of the paper, we will study the measure of nonclassicality of quantum states based on the activation protocol (Definition~\ref{def:quanumnessmeasure}) and on the choice of negativity as entanglement measure, $E \equiv {\cal N}$.

The negativity (of entanglement) is defined as follows.
\begin{defn}[Negativity (of entanglement)~\cite{negativity}]
Let $\rho_{AB}$ be a bipartite quantum state. The \emph{negativity (of entanglement)} of $\rho_{AB}$ is defined as
\[
\neg_{A:B}(\rho_{AB}):=\frac{||\rho_{AB}^\Gamma||_1-1}{2},
\]
where the subscript of $\neg$ denotes the bipartition with respect to which the entanglement is being measured, the superscript~$^\Gamma$ on $\rho_{AB}$ denotes its partial transpose and $\|\cdot\|_1$ is the Schatten 1-norm. This definition assumes, as we do in the rest of the paper, that we are dealing with normalized states.
\end{defn}
It can be immediately verified that the negativity of entanglement is independent both of the choice of the party on which the partial transposition is considered and of the choice of local basis in which the local transposition is taken. The negativity of quantumness can then be defined as follows.
\begin{defn}[Negativity of quantumness (NoQ)~\cite{piani:prl2011a}]
The \emph{negativity of quantumness} (NoQ) of a $n$-partite quantum state $\rho_{\upto{n}}$ on  subsystems $\Sigma\subseteq\upto{n}$ is defined as
\begin{equation}\label{eq:defNoQ}
Q_\neg^\Sigma(\rho_{\upto{n}}):=\min\neg_{\upto{n}:\Sigma'}(\tilde{\rho}_{\Xi}),
\end{equation}
where the minimum is taken over all pre-measurement states $\tilde{\rho}_{\Xi}$, $\Xi=\{1,2,\ldots,n\}\cup\Sigma'$, of the quantum systems $\{1,2,\ldots,n\}$ and the measurement apparatuses $\Sigma'$ associated with the individually measured systems $\Sigma\subseteq \{1,2,\ldots,n\}$.
\end{defn}
Notice that the NoQ has been sometimes referred to as minimum entanglement potential \cite{chavesito}.


\subsection{Total negativity of quantumness}

The total (or two-sided, in the case there are only two subsystems) quantumness (of correlations) of a quantum state is quantified by the amount of apparatus-system entanglement in a pre-measurement state when every subsystem is measured individually, i.e., $\Sigma= \{1,2,\ldots,n\}$. Adopting the NoQ as a measure, the corresponding explicit expression for the total NoQ of arbitrary multipartite states was given in \cite{piani:prl2011a}. Here we will derive this expression again as it will be useful for the analysis made later in this paper.

First, we observe that the pre-measurement states for the case of studying total quantumness have a very particular form.

\begin{defn}[Maximally correlated state (MCS)]
A bipartite state $\tau_{AB}$ of systems $A$ and $B$ is said to have the \emph{maximally correlated form} if it can be expressed as $\tau_{AB}=\sum_{i,j=1}^n\tau_{ij}\state{a_i}{a_j}_A\otimes \state{b_i}{b_j}_B$ for some $\tau_{ij}\in\complex$ with respect to some orthonormal basis of each system $\{|a_i\rangle\}$ and $\{|b_j\rangle\}$ and $n=\min{\{\dim{A},\dim{B}\}}$. Note that Hermiticity implies $\tau_{ij}=\tau_{ji}^*$. A state that can be represented in a maximally correlated form is called \emph{maximally correlated state (MCS)}.
\end{defn}

MCSs have some remarkable properties. E.g.,  if a MCS  has negative partial transpose the entanglement contained in the state is distillable and, moreover, distillable entanglement and relative entropy of entanglement coincide for this set of states~\cite{rains:itit2001a}. Moreover, for any quantumness measure, the quantumness of a MCS is the same as its corresponding entanglement as proven in Appendix~\ref{sec:equality of ent and quant}\footnote{Recently it has been proven that some entanglement measures and some nonclassical correlation measures coincide for a more general class of states \cite{coles:2012a}.}. At the same time, there is a simple analytic expression for both the  eigenvalues and the eigenvectors of its partial transpose.

\begin{lem}
\label{lem:spectrumMCS}
The partial transpose\footnote{Here taken with respect to subsystem $B$, in the maximally correlated basis $\{|b_j\rangle\}$; while the eigenvalues of the partially transposed state do not depend on this details, the eigenvectors do.} $\tau_{AB}^\Gamma=\sum_{i,j=1}^n\tau_{ij}\state{a_i}{a_j}_A\otimes \state{b_j}{b_i}_B$ of a MCS $\tau=\sum_{i,j=1}^n\tau_{ij}\state{a_i}{a_j}_A\otimes \state{b_i}{b_j}_B$ has the following eigenvalues and corresponding eigenvectors:
\begin{align*}
\tau_{ii} \qquad &\text{for}\quad \ket{a_i}\otimes \ket{b_i},
\\
\pm|\tau_{ij}| \quad &\text{for}\quad \frac{1}{\sqrt{2}}\left(\ket{a_i}\otimes \ket{b_j}\pm \frac{\tau_{ji}}{|\tau_{ji}|}\ket{a_j}\otimes \ket{b_i}\right), \quad\text{with }  i> j
\end{align*}
\end{lem}
\begin{proof}
Notice that $\rho_{AB}^{\Gamma}$ with respect to the basis $\{\ket{a_i}\otimes\ket{b_j}\}$ 
 is a generalized permutation matrix, i.e., there is only one nonzero entry on each row and column. Therefore by the action of an appropriate permutation matrix $P$, it can be transformed into the following form (the off-diagonal blocks have all vanishing entries; we do not indicate all the zero entries):
\[
P\microspace\rho_{AB}^{\Gamma} \microspace P^{-1}=
\left(
\begin{BMAT}(rc){c.c.c.c}{c.c.c.c}
 \begin{BMAT}(rc){ccc}{ccc}
        \tau_{11} & &\\
        & \ddots &\\
        & & \tau_{nn}
    \end{BMAT} & & &\\
 & \begin{BMAT}(rc){cc}{cc}
       0& \tau_{12}\\
         \tau_{21}& 0
    \end{BMAT}&&\\
    &&\ddots&\\
    &&&\begin{BMAT}(rc){cc}{cc}
       0& \tau_{n-1,n}\\
         \tau_{n,n-1}& 0
    \end{BMAT}
\end{BMAT}
\right).
\]
Each diagonal block corresponds to an invariant subspace (under the action of $P\microspace\rho_{AB}^{\Gamma} \microspace P^{-1}$). The rows and columns of the upper-left block corresponds to the vectors $\ket{a_i}\otimes \ket{b_i}$'s, while the entries of the other two-by-two on-diagonal blocks each correspond to each $\ket{a_i}\otimes \ket{b_j}$ and $\ket{a_j}\otimes \ket{b_i}$, for $i\neq j$. Hence $P\microspace\rho_{AB}^{\Gamma} \microspace P^{-1}$ has the eigenvalues and eigenvectors above.
\end{proof}

From this we get immediately the following.

\begin{cor}[Negativity of a MCS]\label{cor:negMCS}
The negativity of a MCS $\tau_{AB}=\sum_{i,j=1}^n\tau_{ij}\state{a_i}{a_j}_A\otimes \state{b_i}{b_j}_B$ is
\begin{equation}
\neg(\rho_{AB})=\frac{\sum_{i,j}|\tau_{ij}|-1}{2}.
\end{equation}
\end{cor}

The quantity $\sum_{i,j}|\tau_{ij}|$ in the above corollary  is the sum of the absolute value of entries of the matrix $[\tau_{ij}]$ and it can be seen as the $l_1$-norm of the matrix in the maximally correlated basis (see Section~\ref{sec:notation} and~Appendix~\ref{sec: l1-norm}).  The negativity of a MCS can then be expressed as
\[
\neg(\rho_{AB})=\frac{\lnorm{\tau_{AB}}^{\{|a_i\rangle\otimes|b_j\rangle\}}-1}{2},
\]
where the superscript indicates that the $l_1$-norm is taken for the matrix representation of $\tau_{AB}$ with respect to the basis $\{|a_i\rangle\otimes|b_j\rangle\}$.

We are now ready to prove the following.

\begin{thm}[Total Negativity of Quantumness~\cite{piani:prl2011a}]\label{thm:overallNoQ}
Let $\rho_{\upto{n}}$ be an $n$-partite quantum state. Then the total negativity of quantumness, i.e., the minimum amount of entanglement w.r.t. negativity between system-apparatuses when each subsystem is measured independently, is
\[
Q_\neg^{\upto{n}}(\rho_{\upto{n}})=\min_{\bigotimes_{i\in\upto{n}}\bas{i}}\frac{\lnorm{\rho_{\upto{n}}}^{\bigotimes_{i\in\upto{n}}\bas{i}}-1}{2},
\]
where the minimum is taken over different choices of factorized basis $\bigotimes_{i\in\upto{n}}\bas{i}$ for the (local) measurement interaction.
\end{thm}

\begin{proof}

This result is a direct consequence of every pre-measurement state being a MCS (between system and apparatus) and it is true regardless of the number of systems. For the sake of concreteness we will only prove it for the bipartite case, the extension to the multiparty case being straightforward.

Let $\rho_{AB}$ be a bipartite state and let $\tilde{\rho}_{ABA'B'}$ be its pre-measurement state produced by measurement interactions $V_{A\mapsto AA'}^{\{\ket{a_i}_A\}}$ and $U_{B\mapsto BB'}^{\{\ket{b_i}_B\}}$, where $\bas{A}=\{|a_i\rangle\}$ and $\bas{B}=\{|b_i\rangle\}$ are some orthonormal bases of subsystem $A$ and $B$:
\[
\tilde{\rho}_{ABA'B'}=\sum_{ijkl}\rho_{ijkl}\state{a_i}{a_j}_A\otimes \state{b_k}{b_l}_B\otimes \state{i}{j}_{A'}\otimes \state{k}{l}_{B'},
\]
where $\rho_{ijkl}=\bra{a_i}\bra{b_k}\rho_{AB}\ket{a_j}\ket{b_l}$. Observe that this state has indeed the maximally correlated  form in the $(AB):(A'B')$ cut. Therefore Corollary \autoref{cor:negMCS} implies
\begin{align}
Q_\neg^{AB}(\rho_{AB})&=\min_{\bas{A}\otimes\bas{B}}\frac{\sum_{ijkl}|\rho_{ijkl}|-1}{2}
\\
&=\min_{\bas{A}\otimes\bas{B}}\frac{\lnorm{\rho_{AB}}^{\bas{A}\otimes\bas{B}}-1}{2}.
\end{align}
\end{proof}
\noindent This relation can be rewritten as
\beq
\label{eq:twosidedNoQ}
Q_\neg^{AB}(\rho_{AB})=\min_{\bas{A}\otimes\bas{B}}\frac{1}{2}\lnorm{\rho_{AB}-(\Pi_{\bas{A}}\otimes\Pi_{\bas{B}})[\rho_{AB}]}^{\bas{A}\otimes\bas{B}},
\eeq
where $\Pi_{\bas{A}}$ and $\Pi_{\bas{B}}$ represent complete projective measurements on systems $A$ and $B$, respectively, on the orthonormal basis $\bas{A}$ and $\bas{B}$. The superscript in the norm expression indicates that the $l_1$-norm is taken in the same local basis as for the projective measurement. Therefore the total NoQ can be given the following interpretation.

\begin{cor}[Decoherence interpretation of total NoQ]\label{cor: decoherence interpretation}
The total NoQ of $n$-partite quantum state $\rho_{\upto{n}}$ is the minimum disturbance caused by a complete projective measurement on every system $\upto{n}$ as quantified by the $l_1$-norm in the basis of the measurement.
\end{cor}

Notice that the total NoQ corresponds to the absolute sum of the off-diagonal entries ({\it coherences}) of the density matrix, minimized over all local product bases \cite{piani:prl2011a}.


\subsection{Partial negativity of quantumness}

Here we study the quantumness of correlations due to the nonclassical nature of single subsystems. Notice that this notion of partial (or one-sided, in the case there are only two subsystems) quantumness is well defined since the activation protocol Theorem \ref{thm: activation protocol} applies.

\begin{thm}[Partial negativity of quantumness]
Let $\rho_{\upto{n}}$ be an $n$-partite quantum state and  $\Sigma\subseteq\upto{n}$ a subset of subsystems, and denote the elements of $\Sigma$ as $\{k_1,k_2,\cdots\}$. For each subsystem $k\in\Sigma$, let $\bas{k}=\left\{\ket{a_{i_k}^{(k)}}\right\}_{i_k}$ be one orthonormal basis of $k$.   Then the \emph{partial negativity of quantumness on subsystems $\Sigma\subset\upto{n}$} is
\[
Q_\neg^{\Sigma}(\rho_{\upto{n}})
=
\min_{\bigotimes_{k\in\Sigma}\bas{k}}
\frac{1}{2}
\left(\sum_{i_{k_1},i_{k_2},\ldots} \|\rho_{i_{k_1},i_{k_2},\ldots}\|_1-1\right),
\]
where $\rho_{i_1,i_2,\ldots}=\left(\bra{a^{(k_1)}_{i_1}}\bra{a^{(k_2)}_{i_2}}\cdots\right)\rho_{\upto{n}}\left(\ket{a^{(k_1)}_{i_1}}\ket{a^{(k_2)}_{i_2}}\cdots\right)$ and the minimum is taken over different choices of measurement interaction defined by $\bigotimes_{k\in\Sigma}\bas{k}$.
\end{thm}
\begin{proof}
The same proof applies to a system with an arbitrary number of subsystems, but for the sake of clarity and concreteness only the bipartite case is explicitly treated here. Let $\rho_{AB}$ be a bipartite quantum state; we want to calculate the QoN on subsystem $A$, $Q_\neg^A(\rho_{AB})$. Let $\{a_i\}_{i\in\upto{m}}$ be an orthonormal basis for subsystem $A$, which we assume to have finite dimension $\dim{A}=m$. Suppose that the pre-measurement state
\[
\tilde{\rho}_{ABA'}=\sum_{i,j=1}^m\state{a_i}{a_j}_A\otimes \rho_{ij}\otimes \state{i}{j}_{A'},
\]
where $\rho_{ij}=\rho_{ij}^B=\bra{a_i}\rho_{AB}\ket{a_j}$, is created by a measurement interaction $V_{A\mapsto AA'}^{\{\ket{a_i}_A\}}$. After the action by an appropriate permutation matrix, the partially transposed $\tilde{\rho}_{ABA'}$ on $A'$ can be written in a block-diagonal form as we did for the maximally correlated states in the proof of Lemma~\ref{lem:spectrumMCS}:
\[
P\microspace\tilde{\rho}_{ABA'}^{\Gamma} \microspace P^{-1}=
\left(
\begin{BMAT}(rc){c.c.c.c}{c.c.c.c}
 \begin{BMAT}(rc){ccc}{ccc}
        \rho_{11} & &\\
        & \ddots &\\
        & & \rho_{mm}
    \end{BMAT} & & &\\
 & \begin{BMAT}(rc){cc}{cc}
       0& \rho_{12}\\
         \rho_{21}& 0
    \end{BMAT}&&\\
    &&\ddots&\\
    &&&\begin{BMAT}(rc){cc}{cc}
       0& \rho_{m-1,m}\\
         \rho_{m,m-1}& 0
    \end{BMAT}
\end{BMAT}
\right).
\]
The only difference with the MCS case is that here the entries of the block matrices on the diagonal are matrices themselves rather than  scalars. It is well known that the eigenvalues of a Hermitian matrix  $\begin{pmatrix}
 0 & H\\
H^\dagger & 0 \\
\end{pmatrix}$ for an arbitrary matrix $H$ are given by the singular values of $H$, taken with both positive and negative sign. Therefore, while the Schatten $1$-norm of the upper-left block is clearly $\sum_{i=1}^{m}\|\rho_{ii}\|_1$,  that of the second block diagonal matrix is $\|\rho_{1,2}\|_1+\|\rho_{2,1}\|_1$, etc. Now, the Schatten $1$-norm of block diagonal matrices is simply the sum of Schatten $1$-norms of each sub-block. Therefore
\begin{equation}\label{eq:onesidedNoQ}
Q_\neg^A(\rho_{AB})=\min_{\bas{A}}\frac{1}{2}\left(\sum_{i,j}||\rho_{ij}||_1-1\right).
\end{equation}
\end{proof}


\section{Properties and interpretations of total and partial negativity of quantumness}\label{sec:interp}
\subsection{Properties of NoQ}
We list here some general properties of NoQ:
\begin{itemize}
\item
\emph{positivity}: NoQ is nonnegative for any quantum state;
\item
\emph{faithfulness}: NoQ  is faithful, i.e., it is zero if and only if the state is classical on the subsystems that are measured (all the subsystems in the case of the total NoQ);
\item
\emph{negativity of quantumness exceeds negativity of entanglement~\cite{piani:pr2012a}}: for any bipartite quantum state, the total NoQ exceeds the partial NoQ and both are always larger than the entanglement of the state as quantified by the negativity.
\end{itemize}
The first two properties are very natural requirements for a measure of nonclassical correlations. The first property follows directly from the definition; the second property follows from the activation protocol of Theorem \ref{thm: activation protocol} and the fact that for a MCS, negativity is nonzero if and only if the state is entangled~\cite{rains:itit2001a}. 
The third property follows from the hierarchy theorem of \cite{piani:pr2012a} and is exploited in the next subsection.

\subsection{Negativity of entanglement versus negativity of quantumness}\label{sec:NvsQN}

In \cite{khasin:pr2007a}, Khasin {\it et al.} studied the negativity of an arbitrary MCS and made a similar observation as ours---that is to say, the negativity of a MCS can be geometrically interpreted as the $l_1$-norm distance from the separable state given by considering only the diagonal component of that MCS. Using our notation, their result is the following.
\begin{thm}[\cite{khasin:pr2007a}]
Let $\tau_{AB}=\sum_{i,j=1}^n\tau_{ij}\state{a_i}{a_j}_A\otimes \state{b_i}{b_j}_B$  be a MCS for some $\tau_{ij}\in\complex$ with respect to some orthonormal basis of each system $\bas{A}=\{|a_i\rangle\}$ and $\bas{B}=\{|b_j\rangle\}$. Then,
\begin{equation}\label{eq: khasin}
\neg_{A:B}(\tau_{AB})=\frac{1}{2}\lnorm{\tau_{AB}-\Pi_{\bas{A}}\otimes\Pi_{\bas{B}}[\tau_{AB}]}^{\bas{A}\otimes\bas{B}}
\end{equation}
where $\Pi_A,\Pi_B$ are complete projective measurements on the basis $\bas{A},\bas{B}$.
\end{thm}
\noindent Observe that we have already obtained the same result, thanks to Eq.\eqref{eq:twosidedNoQ}  and the equality of entanglement and quantumness for a MCS (see Appendix~\ref{sec:equality of ent and quant} or Ref.~\cite{coles:2012a}).

Besides such a result, Khasin {\it et al.}~conjectured that the negativity of any bipartite state is upper-bounded by the above quantity minimized over the choice of local bases.  In this section, we prove their conjecture and moreover prove that for a MCS $\tau_{AB}$ the state $(\Pi_A\otimes\Pi_B)[\tau_{AB}]$  in Eq.\eqref{eq: khasin} is indeed the closest separable state of $\tau_{AB}$ with respect to the $l_1$-norm in the ${\bas{A}\otimes\bas{B}}$ maximally correlated basis.

With respect to the conjecture of Khasin {\it et al.}, we find the following.
\begin{thm}
For any bipartite state $\rho_{AB}$ it holds
\[
\neg(\rho_{AB})\leq \min_{\bas{A}\otimes\bas{B}}\frac{1}{2}\lnorm{\rho_{AB}-\Pi_{\bas{A}}\otimes \Pi_{\bas{B}}[\rho_{AB}]}^{\bas{A}\otimes\bas{B}},
\]
where the minimization is taken over the choice of local product bases $\bas{A}\otimes\bas{B}$, and $\Pi_A$ is a complete projection onto $\bas{A}$ (similarly for $B$).
\end{thm}
\begin{proof}
We know from \cite{piani:pr2012a} that $\neg(\rho_{AB})\leq Q_{\neg}(\rho_{AB})$ holds for all $\rho_{AB}$. The claim is then obtained combining this with Eq.~\eqref{eq:twosidedNoQ}.
\end{proof}

We are also able to prove the following.
\begin{thm}\label{thm: closestclassicalMCS}
Consider a MCS $\tau_{AB}=\sum_{ij}\tau_{ij}|a_i\rangle\langle a_j|\otimes |b_i\rangle\langle b_j|$ where $\bas{A}=\{\ket{a_i}\}$ and $\bas{B}=\{\ket{b_j}\}$ are any orthonormal bases of subsystem $A$ and $B$. Let us define a state $\sigma_{AB}=\sum_{i}\tau_{ii}|a_i\rangle\langle a_i|\otimes |b_i\rangle\langle b_i|$, i.e., $\sigma$ only has the diagonal components of $\tau_{AB}$ in the chosen local bases.  Then one of the closest separable states (or, more precisely, one of the closest PPT states) to $\tau_{AB}$ with respect to the $||\cdot||_{l_1}$ norm in the $\bas{A}\otimes\bas{B}$ basis is $\sigma_{AB}$:
\begin{align*}
 \lnorm{\tau_{AB}-\sigma_{AB}}^{\bas{A}\otimes\bas{B}}&=\min_{\xi\in\ppt}\lnorm{\tau_{AB}-\xi}^{\bas{A}\otimes\bas{B}}\\
&=\min_{\xi\in\sep}\lnorm{\tau_{AB}-\xi}^{\bas{A}\otimes\bas{B}},
\end{align*}
where $\ppt$ is the set of all $AB$ states with positive partial transpose, $\sep$($\subset\ppt$) is the set of all separable state in $AB$,  and the $l_1$-norm is taken with respect to the basis $\bas{A}\otimes\bas{B}$ in which $\tau_{AB}$ has the maximally correlated form.
\end{thm}
\begin{proof} See Appendix~\ref{sec: proof}.
\end{proof}

\subsection{Interpretations of the negativity of quantumness}
Any measure of quantumness defined through the activation protocol naturally possesses an operational meaning---the least amount of system-apparatus entanglement which will be created by any measurement interaction. Namely, NoQ quantifies such minimum entanglement in terms of negativity. However it turns out that NoQ has some more possible interpretations: a geometric interpretation as the minimum distance from classical states (\ref{itm:distance}), and an operational interpretation in terms of disturbance induced by a measurement (\ref{itm:decoherence}).
Here we restrict our attention to the study of bipartite states $\rho_{AB}$. In the bipartite case, we often refer to he total (partial) NoQ  as to the two-sided (one-sided) NoQ.

With respect to the operational interpretation in terms of measurement disturbance, when the decohered quantum system $A$ is a qubit the one-sided NoQ can be interpreted as the distinguishability of a quantum state from its partially decohered state.


\begin{thm}
Let $\rho_{AB}$ be a bipartite quantum state with A a qubit. Then the partial NoQ on subsystem $A$ of quantum state $\rho_{AB}$ is equivalent to the minimum trace distance between $\rho_{AB}$ and its decohered state on subsystem $A$:
\begin{equation}\label{eq:onesideinter}
Q_\neg^A(\rho_{AB})=\min_{\bas{A}}\frac{1}{2}||\rho_{AB}-\Pi_{\bas{A}}[\rho_{AB}]||_1,
\end{equation}
where $\Pi_{\bas{A}}$ is the complete projective measurement in the basis $\bas{A}$ and the minimum is taken for different choices of basis $\bas{A}$.
\end{thm}
\begin{proof}
Suppose the subsystem $A$ is a two-level system, so that the sum in Eq.\eqref{eq:onesidedNoQ} is limited to $i,j\in\{0,1\}$. Then,
\begin{align*}
\sum_{i,j=0}^1||\rho_{ij}||_1&= \|\rho_{00}\|_1 + \|\rho_{11}\|_1 + \|\rho_{01}\|_1 + \|\rho_{10}\|_1\\
					&= 1 + \left\| 	\begin{pmatrix}
 									0 & \rho_{01}\\
									\rho_{10} & 0 \\
								\end{pmatrix}
						\right\|_1\\
					&=1 + \| \rho_{AB} -  \Pi_{\bas{A}}[\rho_{AB}]| \|_1.
\end{align*}
Here, the second equality is due on one side to the fact that, $\rho_{AB}$ being a normalized state, $\|\rho_{00}\|_1 + \|\rho_{11}\|_1 = 1$; on the on the other side, for a block anti-diagonal matrix one has $ \left\| \begin{pmatrix}
 0 & X\\
 Y & 0 \\
\end{pmatrix} \right\|_1 = \|X\|_1 +\|Y\|_1$.

\end{proof}

Furthermore, we  find the following equivalence.

\begin{thm}
Let $\rho_{AB}$ be a bipartite quantum state where the quantum system A is two-dimensional. Then
\beq
\label{eq:onesideequiv}
\min_{\Pi_{\bas{A}}}\|\rho_{AB}-\Pi_{\bas{A}}[\rho_{AB}]\|_1=\min_{\sigma\in\CQ}\|\rho_{AB}-\sigma\|_1.
\eeq
\end{thm}

\begin{proof}
The inequality
\[
\min_{\Pi_{\bas{A}}}\|\rho_{AB}-{\Pi_{\bas{A}}}[\rho_{AB}]\|_1 \ge \min_{\sigma\in\CQ}\|\rho_{AB}-\sigma\|_1
\]
 holds trivially. Therefore we will prove the other direction.

Let $\{\ket{i}\}_{\{i=0,1\}}$ be the orthonormal basis of $A$ in which the minimization of right-hand side of Eq.~\eqref{eq:onesideequiv} is achieved (i.e., in which $\sigma$ is explicitly classical on $A$, for an optimal classical-quantum $\sigma$). Now, $\rho-\sigma$ in the chosen basis has the following block-matrix form:
\[
\begin{pmatrix} A & \rho_{01} \\ \rho_{10} & B \end{pmatrix},
\]
with $\rho_{AB}=\begin{pmatrix} \rho_{00} & \rho_{01} \\ \rho_{10} & \rho_{11} \end{pmatrix}$,  $\sigma_{AB}=\begin{pmatrix} \sigma_{00} & 0 \\ 0 & \sigma_{11} \end{pmatrix}$, $A= \rho_{00} -\sigma_{00}$, and $B=\rho_{11} -\sigma_{11}$. 
Let the singular value decomposition of $\rho_{01}$ be $\rho_{01}=UDV$ and note that
\[
\begin{split}
\left\|\begin{pmatrix} A & \rho_{12} \\ \rho_{21} & B \end{pmatrix}\right\|_1
&= \left\| \begin{pmatrix} U & 0\\ 0 & V^\dagger \end{pmatrix}\begin{pmatrix} U^\dagger AU & D \\ D & VBV^\dagger \end{pmatrix}\begin{pmatrix} U^\dagger & 0\\ 0 & V \end{pmatrix}\right\|_1\\	
&=\left\| \begin{pmatrix} U^\dagger AU & D \\ D & VBV^\dagger \end{pmatrix}\right\|_1\\
&\geq \left\| \begin{pmatrix} \Pi[U^\dagger AU] & D \\ D & \Pi[VBV^\dagger] \end{pmatrix}\right\|_1 \\
&= \left\| \begin{pmatrix} D_1 & D \\ D & D_2 \end{pmatrix}\right\|_1
\end{split}
\]
where $\Pi$ is a complete projective measurement on $B$ that leaves $D$ invariant, $D_1=\Pi[U^\dagger AU] $, $D_2=\Pi[VBV^\dagger]$, and we have used that the trace distance is invariant under unitaries and monotone under general quantum operations.


Let $D_1=\diag{\{a_1,\cdots,a_n\}}$, $D_2=\diag{\{b_1,\cdots,b_n\}}$, and $D=\diag{\{c_1,\cdots,c_n\}}$. Notice that $D$ depends only on the state $\rho$, not on $\sigma$ nor on the choice of basis.  Then
\begin{align*}
\|\rho_{AB}-\sigma\|_1 &\geq \left\| \begin{pmatrix} D_1 & D \\ D & D_2 \end{pmatrix}\right\|_1
\\
&=\sum_{i=1}^n \left\| \begin{pmatrix} a_i & c_i \\ c_i & b_i \end{pmatrix}\right\|_1.
\end{align*}
Now
\[ 
\left\| \begin{pmatrix} a & c \\ c & b \end{pmatrix}\right\|_1=\max_{U}\left|\Tr\left( U \begin{pmatrix} a & c \\ c & b \end{pmatrix}\right)\right|\geq \left|\Tr\left(  \begin{pmatrix} 0 & 1 \\ 1 & 0 \end{pmatrix} \begin{pmatrix} a & c \\ c & b \end{pmatrix}\right)\right|= 2 c ,
\] 
hence
\[
\min_{\sigma\in\CQ}\|\rho_{AB}-\sigma\|_1
%
\ge \sum_i 2 c_i,
\]
and such lower bound can be achieved when $a_i=b_i=0$ for all $i$, i.e., when $A=B=0$. The latter case corresponds to $\sigma_{00}=\rho_{00} $ and  $\sigma_{11}=\rho_{11} $, i.e., $\sigma= {\Pi_{\bas{A}}}[\rho_{AB}]$.
\end{proof}

Hence, when the single measured system is a qubit, \ref{itm:distance} and \ref{itm:decoherence} are equivalent also in the case in which the distance adopted is the {\it trace distance} rather than the Hilbert-Schmidt distance used in the case of the `standard' geometric discord~\cite{dakic:prl2010a,luo2010geometric}. Combining Eqs. \eqref{eq:onesideinter} and \eqref{eq:onesideequiv} with the hierarchy of \cite{piani:pr2012a}, we have that the trace-distance based quantumness measure always exceeds the negativity (of entanglement) in bipartite states $\rho_{AB}$ where $A$ is a qubit. We have numerical evidence that the equivalence between the three approaches---activation (with negativity), geometric (with trace distance), and disturbance (with trace distance)---breaks down when the measured system is not a qubit.

For the two-sided case, we have already seen in Corollary \ref{cor: decoherence interpretation} that in general the total NoQ can be interpreted as the minimum disturbance caused by local complete projective measurements on $A$ and $B$, as quantified by the $l_1$-norm in the basis in which the projection is taken. Moreover, by choosing properly the distance, it can be also interpreted from the perspective of \ref{itm:distance}.

\begin{thm}\label{thm:twoNoQinter}
Let $\rho_{AB}$ be a bipartite state. Then its two-sided NoQ is equivalent to the distance in $l_1$-norm from its closest $\CC$-state where the norm is with respect to the eigenbasis of the classical state.
\begin{equation}
Q_\neg^{AB}(\rho_{AB})=\min_{\sigma\in\CC}\frac{1}{2}\lnorm{\rho_{AB}-\sigma}^{\bas{\sigma}},
\end{equation}
where $\bas{\sigma}$ denotes the eigenbasis\footnote{An eigenbasis is uniquely defined if there are no degeneracies in the spectrum; if there are, it is implicit that the basis $\bas{\sigma}$ is chosen optimally to minimize the distance.} of $\sigma$.
\end{thm}
\begin{proof}
Consider having fixed the local bases $\bas{A}$ and $\bas{B}$ for the $\CC$-state $\sigma$, so that $\bas{\sigma}=\bas{A}\otimes\bas{B}$.  Optimize now over the eigenvalues of $\sigma$, for fixed $\bas{\sigma}$. It is clear that within such a class, the  $\CC$-state $\sigma$ that is optimal for the sake of $\lnorm{\rho_{AB}-\sigma}^{\bas{\sigma}}$ is the one that has the same diagonal (in the fixed basis) entries as $\rho$. For such $\sigma$ it holds $\lnorm{\rho_{AB}-\sigma}^{\bas{\sigma}}=\lnorm{\rho_{AB}-(\Pi_{\bas{A}}\otimes\Pi_{\bas{B}})[\rho_{AB}]}^{\bas{A}\otimes\bas{B}}$.  The remaining minimization over the choice of local bases is the same as in Eq.\eqref{eq:twosidedNoQ}.
\end{proof}
 
\subsection{The mechanism of the activation protocol for negativity}

In this subsection, we will see that the (total) activation protocol using as entanglement measure the negativity can be described in terms of an isometric mapping due to the measurement interaction: 
a bipartite quantum state and its closest $\CC$-state are mapped to a pre-measurement state and the separable state closest to the latter, respectively. 

The following lemma makes it clear what we mean by the fact that the  measurement interaction is isometric with respect to the $l_1$-norm.

\begin{lem}
The $l_1$-norm is invariant under a measurement interaction provided that the $l_1$-norm is taken in the basis in which the measurement interaction is defined.
\end{lem}
\begin{proof} The proof is immediate and for the sake of clarity and concreteness we will consider only the case of the measurement of a single system.
Consider an operator $X$ on system $A$ and the measurement interaction $V_{A\mapsto AA'}^{\{\ket{a_i}_A\}}$ that acts on an orthonormal basis $\{\ket{a_i}\}$of $A$  as
\[
\ket{a_i}_A\mapsto\ket{a_i}_A\ket{i}_{A'}.
\]
Then the matrix representations of $X$ in the basis $\{\ket{a_i}\} $ and that of $V_{A\mapsto AA'}^{\{\ket{a_i}_A\}}X(V_{A\mapsto AA'}^{\{\ket{a_i}_A\}})^\dagger$ in the basis $\{\ket{a_i}_A\otimes\ket{i}_{A'}\}$ have nonzero terms that are in one-to-one correspondence, so that
\[
\lnorm{X}^{\{\ket{a_i}_A\}}=\lnorm{V_{A\mapsto AA'}^{\{\ket{a_i}_A\}}X(V_{A\mapsto AA'}^{\{\ket{a_i}_A\}})^\dagger}^{\{\ket{a_i}_A\otimes\ket{i}_{A'}\}}
\]
holds.
\end{proof}

This isometric property of the measurement interaction and the following two observations---actually, restatements of results we obtained in the previous sections---allow us to draw a clear picture of the activation mechanism.

The first observation is that, according to Theorem \autoref{thm:twoNoQinter}, the total NoQ of a bipartite quantum state can be interpreted as the $l_1$-norm distance from its closest classical state:
\begin{align}
Q_\neg^{AB}(\rho_{AB})&=\min_{\sigma\in\CC}\frac{1}{2}\lnorm{\rho_{AB}-\sigma}^{\bas{\sigma}}\nonumber\\
&=\min_{\bas{A}\otimes\bas{B}}\frac{1}{2}\lnorm{\rho_{AB}-(\Pi_{\bas{A}}\otimes\Pi_{\bas{B}})[\rho_{AB}]}^{\bas{A}\otimes\bas{B}}. \nonumber
\end{align}

The second observation is that, according to Theorem \autoref{thm: closestclassicalMCS}, the closest separable state to a MCS $\tau_{AB}$ is the separable state corresponding to the diagonal---in the maximally correlated basis---entries of $\tau_{AB}$:
\[
\min_{\xi\in\ppt}\lnorm{\tau_{AB}-\xi}^{\bas{A}\otimes\bas{B}}=\lnorm{\tau_{AB}-\sigma}^{\bas{A}\otimes\bas{B}}.
\]

Now, the measurement interaction acting with respect to the basis defined by $\Pi_{\bas{A}}$ and $\Pi_{\bas{B}}$ acts on the state and on its closest classical state as follows:
\begin{align}
\rho_{AB}&\mapsto\tilde{\rho}_{ABA'B'}\nonumber\\
(\Pi_A\otimes\Pi_B)[\rho_{AB}]&\mapsto (\Pi_A\otimes\Pi_B)[\tilde{\rho}_{ABA'B'}].\nonumber
\end{align}
Due to Theorem~\ref{thm: closestclassicalMCS}, the state $(\Pi_A\otimes\Pi_B)[\tilde{\rho}_{ABA'B'}]$ is indeed the closest separable state to the pre-measurement state $\tilde{\rho}_{ABA'B'}$. In Appendix~\ref{sec:isometric for rel} we argue that that same isometric mapping ``state~$\mapsto$~pre-measurement state'' and ``closest classical state~$\mapsto$~closest (to the pre-measurement state) separable state'' holds also for the case of relative entropy used as a distance function.

The following diagram shows the isometric mapping for the NoQ.
\begin{widetext}
\begin{equation}\label{diagrammauno}
\begin{diagram}[height=2em,width=0.3em]
$${\displaystyle{\min_{{\Pi_{A}\otimes\Pi_{B}}}}\lnorm{\tilde{\rho}_{{ABA'B'}}-(\Pi_A\otimes\Pi_B)[\tilde{\rho}_{{ABA'B'}}]}^{{\Pi_{A}\otimes\Pi_{B}}}}$$ &=^{(\romannumeral 2)} & $${\displaystyle{\min_{{\eta\in\sep}}}\lnorm{\tilde{\rho}_{ABA'B'}-\eta}^{{\MCS}}}$$\\
=_{(\romannumeral 1)} & &=\neg(\tilde{\rho}_{{ABA'B'}})\\
$${\displaystyle{\min_{{\Pi_A\otimes\Pi_B}}}\lnorm{\rho_{AB}-(\Pi_A\otimes\Pi_B)[\rho_{AB}]}^{{\Pi_A\otimes\Pi_B}}}$$ & =_{(\romannumeral 3)}&  $$\>{\displaystyle{\min_{\sigma\in\CC}}\lnorm{\rho_{AB}-\sigma}^{{\CC}}}$$\\
&&=Q_{\neg}(\rho_{AB})
\end{diagram}
\end{equation}
\end{widetext}
In Eq.~\eqref{diagrammauno}: Equality (\rmnum{1}) holds because the measurement interaction for the basis $\Pi_A\otimes\Pi_B$ is isometric; Relation (\rmnum{2}) corresponds to the fact that the closest separable state to a MCS is given by its diagonal part; Relation (\rmnum{3}) holds because the closest $\CC$ state is again its diagonal part. With $\lnorm{\cdot}^{\MCS}$ we indicate that the $l_1$ norm is taken in the MCS basis of $\tilde{\rho}_{ABA'B'}$; similarly, with $\lnorm{\cdot}^{{\CC}}$ we indicate that the $l_1$ norm is taken in the basis in which $\sigma$ is excplicitly $\CC$.


\section{The equivalence of nonclassicality measures for bipartite systems with a two-level subsystem}\label{sec: qubit system}

Recently, S. Gharibian has shown in~\cite{gharibian:2012a} that the classicality of a bipartite quantum state $\rho_{AB}$ on subsystem $A$ can be tested by the verifying the invariance of the state under some special local unitary operations. Similar results have been obtained independently by Giampaolo {\it et al.} \cite{faberdagmar}. Both works, which define the \ref{itm:localunitary} in~Sec.~\ref{sec:introduction}, are based on a generalization to mixed states~\cite{gerrylu} of an approach to the quantification of pure-state entanglement via local unitaries \cite{giampaolo-illuminati}.
In~\cite{gharibian:2012a,faberdagmar} it was proven that a quantum state is classical on one subsystem $A$ if and only if there exists some operation from a particular set of local unitaries acting on the subsystem $A$, called {\em Root-of-Unity} operations, that leaves the state invariant. Therefore the minimum disturbance caused by a local unitary from this nontrivial set was suggested as a measure of nonclassicality of correlations. We will show that if the subsystem $A$ under investigation is a two-level system (qubit), then the corresponding measure of nonclassicality defined by this approach (\ref{itm:localunitary}), and those defined by \ref{itm:activation} and \ref{itm:decoherence}, are all related\footnote{The analytic equivalence of geometric discord and nonclassicality measure by local unitary invariance for a $2\times N$-dimensional system was shown in~\cite{gharibian:2012a,faberdagmar} for the special case of Hilbert-Schmidt distance, but not from our general operational viewpoint.}.

\begin{defn}[Root-of-Unity operation~\cite{gharibian:2012a,faberdagmar,gerrylu}]
Consider a $n$-dimensional quantum system $A$. Then the set of all unitary operators on $A$ with spectrum $\{\omega_n^j\}_{j\in\{0,1,\cdots,n-1\}}$ for $\omega_n=e^{2\pi i/n}$ is called the set of {\em Root-of-Unity (RU)} operations~\footnote{Note that this set correspond to  $\{UZU^\dagger\,|\,U \textrm{ unitary}, Z=\sum_j\omega^j_n \ket{j}\bra{j} \}$ for $\{\ket{j}\}$ fixed to be, e.g., the computational basis.}. We indicate such a set by $RU(A)$.
\end{defn}

\ref{itm:localunitary} is then expressed by the following theorem.

\begin{thm}[Non-classicality by local unitary disturbance~\cite{gharibian:2012a,faberdagmar}]
A bipartite quantum state $\rho_{AB}$ is classical on subsystem $A$ if and only if there exists a local unitary operation $V_A\in RU(A)$ which leaves the quantum state invariant, i.e.,
\[
(V_A\otimes \I_B)\rho_{AB}(V_A\otimes \I_B)^\dagger=\rho_{AB}.
\]
\end{thm}

Now consider the case in which $A$ is a qubit. By definition, a root-of-unitary $V_A\equiv V$ of the qubit $A$ has eigenvalues $\pm1$ and thus its spectral decomposition can be expressed as \[
V=\state{\phi}{\phi}-\state{\phi^\perp}{\phi^\perp},
\]
where $\{\ket{\phi},|\phi^\perp\rangle\}$ is an orthonormal basis of a qubit. Then it is easy to see that the mapping
\[
\rho\mapsto \frac{1}{2}(\rho+V\rho V^\dagger)
\]
corresponds the totally dephasing operation in the basis $\{\ket{\phi},|\phi^\perp\rangle\}$.
This immediately implies
\[
\rho-\Pi_A^{\ket{\phi}}[\rho]=\frac{1}{2}(\rho-V\rho V^\dagger)
\]
where $\Pi_A^{\ket{\phi}}$ denotes the totally dephasing operation on system $A$ in the basis $\{\ket{\phi},|\phi^\perp\rangle\}$. Therefore when $A$ is a qubit, the quantification of nonclassicality of correlations by the two different approaches~\ref{itm:decoherence} and~\ref{itm:localunitary} is equivalent up to a constant.
In the original papers~\cite{gharibian:2012a,faberdagmar,gerrylu}, the minimum disturbance caused by a RU operation was measured by the Hilbert-Schmidt norm, but in principle any norm-based distance can be considered, with the equivalence staying true.


\section{Analytic examples}\label{sec:examples}

In this section, we will look at some special classes of states and obtain an analytic expression for the NoQ of these states. Indeed, as we have seen, the expression of NoQ includes an optimization over local bases and finding a closed analytic expression for general states appears to be challenging. The classes of states we study here, two-qubit states with a maximally mixed marginal, Werner states and isotropic states, all have properties---in particular, symmetries---that allow us to simplify the optimization. In \cite{piani:prl2011a} it was already proven that for two qudits
\[
Q_{\mathcal{N}}^{A}\left((1-p)\frac{\openone}{d^2}+p \proj{\psi}\right)=p\mathcal{N}(\proj{\psi}),
\]
where $p$ is a probability, $\openone/d^2$ is the maximally mixed state of the two qudits and $\ket{\psi}$ and arbitrary pure state. This previous result encompasses also two-qubit states not considered here, as well as isotropic states in arbitrary dimensions. In the latter case, the proof provided below is simpler and based on symmetry considerations.


\subsection{Two-qubit states with one maximally mixed marginal}
 
\subsubsection{Bell diagonal states}

Bell states are maximally entangled two-qubit states of the following form:
\begin{align*}
\ket{\phi^\pm} &= \frac{1}{\sqrt{2}}(\ket{0}\ket{0}\pm\ket{1}\ket{1}),
\\
\ket{\psi^\pm} &= \frac{1}{\sqrt{2}}(\ket{0}\ket{1}\pm\ket{1}\ket{0}).
\end{align*}
Bell diagonal states are two-qubit states such that all of their eigenvectors are Bell states, i.e., of the form
\[
\tau=p_0\ket{\phi^+}\bra{\phi^+}+p_1\ket{\psi^+}\bra{\psi^+}+p_2\ket{\psi^-}\bra{\psi^-}+p_3\ket{\phi^-}\bra{\phi^-},
\]
with $\{p_i\}$ a probability vector.  The properties of Bell diagonal states that are used for the proof of the theorems in this section are summarized in Appendix~\ref{sec: Bell and Pauli}. Most importantly, Bell diagonal states can be completely characterized by three alternative---i.e., besides the above probability vector---real parameters, namely by the three elements $R_{11}$, $R_{22}$ and $R_{33}$ of the the \emph{correlation matrix}
$\bold{R}=[R_{\mu\nu}]$, with
\[
R_{\mu\nu}=\tr[(\sigma_\mu\otimes\sigma_\nu)\rho], \quad\mu,\nu=0,1,\ldots,3.
\]

We are able to compute explicitly the one-sided NoQ for Bell diagonal states\footnote{Similar results have been obtained independently by V. Giovannetti with a different method \cite{giovannettiprep}.}. Our result is summarized in the following theorem.
\begin{thm}[One-sided NoQ for Bell diagonal states]
\label{thm:onesidedbell}
Let $R_{00},R_{11},R_{22},R_{33}$ be the correlation matrix elements of a Bell diagonal state $\rho_{AB}$. Rename and reorder  $|R_{11}|,|R_{22}|,|R_{33}|$ according to their size as $\lambda_1, \lambda_2,\lambda_3$ with $1\geq \lambda_3\geq \lambda_2\geq \lambda_1\geq 0$. Then 
\begin{equation}\label{eq:noqindovino}
Q_{\mathcal{N}}^{A}(\rho_{AB})=\frac{\lambda_2}{2}.\end{equation}
\end{thm}
\begin{proof}
Let $\ket{\phi}=(\ket{0}_A\ket{0}_B+\ket{1}_A\ket{1}_B)/\sqrt{2}$ and $\phi_{AB}=|\phi\rangle\langle\phi|$, where throughout the proof $\ket{i}$ denotes the computational basis. Then, as explained in the Appendix~\ref{sec: Bell and Pauli}, a Bell diagonal state can be expressed as $(\mathbb{I}_A\otimes\Lambda_B)\lbrack\phi_{AB}\rbrack$ where $\Lambda(X)=\sum_{\mu=0}^3p_\mu \sigma_\mu X\sigma_\mu$ is a Pauli channel with $\{p_\mu\}_{\mu\in\{0,1,2,3\}}$ a probability vector. Now let $\{\ket{a_i}:i=0,1\}$ be a basis for the qubit space. Then the expression of one-sided NoQ in Eq.~\eqref{eq:onesidedNoQ} implies
\begin{equation}
\label{eq:bellfirststeps}
\begin{split}
Q_{\mathcal{N}}^{A}(\rho_{AB})& = \min_{\lbrace\ket{a_0},\ket{a_1}\rbrace} \big(\|\langle a_0|(\mathbb{I}_A\otimes\Lambda_B\lbrack\phi_{AB}\rbrack)|a_1\rangle\|_1\\ & \quad +\|\langle a_1|(\mathbb{I}_A\otimes\Lambda_B\lbrack\phi_{AB}\rbrack)|a_0\rangle\|_1\big)/2
\\
&\stackrel{(i)}{=} \min_{\lbrace\ket{a_0},\ket{a_1}\rbrace}\max_U |\langle U,\, \langle a_0|(\mathbb{I}_A\otimes\Lambda_B\lbrack\phi_{AB}\rbrack)|a_1\rangle\rangle|
\\
& \stackrel{(ii)}{=} \min_{\lbrace\ket{a_0},\ket{a_1}\rbrace}\max_U |\langle U,\, \Lambda_B\lbrack\langle a_0|\phi_{AB}|a_1\rangle\rbrack\rangle|
\\
& \stackrel{(iii)}{=} \min_{\lbrace\ket{a_0},\ket{a_1}\rbrace}\max_U |\langle \Lambda_B\lbrack U\rbrack,\, \langle a_0|\phi_{AB}|a_1\rangle\rangle|
\\
&\stackrel{(iv)}{=} \min_{\lbrace\ket{a_0},\ket{a_1}\}}\max_U |\langle \Lambda_B\lbrack U\rbrack,\, (\ket{a_0^*}\bra{a_1^*}_B/2)\rangle|
\\
&\stackrel{(v)}{=} \min_{V}\max_U |\bra{1}_BV^\dagger\Lambda_B\lbrack U\rbrack V\ket{0}_B|/2
\end{split}
\end{equation}
where: $(i)$ holds because of the 1-norm invariance under Hermitian conjugation, so that the two terms on the right-hand side of the previous line are equal; $(ii)$ holds because tracing on $A$ commutes with operations on $B$; $(iii)$ holds because a Pauli channel is self-dual; $(iv)$ holds because for the maximally entangled state $\ket{\phi}$ one has $\bra{\gamma}_A\ket{\phi}_{AB}=\ket{\gamma^*}_B/\sqrt{2}$, where $\ket{\gamma^*}=\gamma_0^*\ket{0} + \gamma_1^*\ket{1}$ if $\ket{\gamma}=\gamma_0\ket{0} + \gamma_1\ket{1}$; finally, $(v)$ holds because $\{\ket{a_i^*}:i=0,1\}$ is still an orthonormal basis and we can write $\ket{a_i^*}=V\ket{i}$ for $V$ a unitary.
Since $\bra{1}_BV^\dagger\Lambda_B\lbrack U\rbrack V\ket{0}_B|/2$ is an expression only on subsystem $B$, the subscripts will be omitted in the following.

Now, the spectral decomposition of the unitary $U$ can be expressed as
\begin{align*}
U &= e^{i\theta_1}|\psi\rangle\langle\psi|+e^{i\theta_2}|\psi^\perp\rangle\langle\psi^\perp|
\\
&= e^{i\theta_1}(|\psi\rangle\langle\psi|+e^{i(\theta_2-\theta_1)}(\mathbb{I}-|\psi\rangle\langle\psi|))
\\
&= e^{i\theta_1}(e^{i\theta}\mathbb{I}+(1-e^{i\theta})|\psi\rangle\langle\psi|)
\end{align*}
where $|\psi\rangle,|\psi^\perp\rangle$ are eigenvectors of $U$ and $\theta=\theta_2-\theta_1$.
Thus, ignoring the irrelevant global phase, we can rewrite the NoQ as
\begin{equation}
\label{eq:bellsecondsteps}
\begin{split}
Q_{\mathcal{N}}^{A}(\rho_{AB})&=  \min_{V}\max_{\theta,|\psi\rangle} |\langle1|V^\dagger\Lambda( e^{i\theta}\mathbb{I}+(1-e^{i\theta})|\psi\rangle\langle\psi|) V|0\rangle|/2
\\
&\stackrel{(i)}{=} \min_{V}\max_{|\psi\rangle} |\langle1|V^\dagger\Lambda\lbrack|\psi\rangle\langle\psi|\rbrack V|0\rangle|
\\
&\stackrel{(ii)}{=} \min_{V}\max_{|\psi\rangle}  |\langle \frac{\sigma_x+i\sigma_y}{2},\, V^\dagger\Lambda\lbrack|\psi\rangle\langle\psi|\rbrack V\rangle|
\\
&= \min_{V}\max_{|\psi\rangle}  \sqrt{\langle\sigma_x\rangle^2+\langle\sigma_y\rangle^2}/2
\end{split}
\end{equation}
where $\langle\sigma_x\rangle$ is the expectation value of $\sigma_x$ for the state $V^\dagger\Lambda(|\psi\rangle\langle\psi|)V$, i.e., $\langle \sigma_x,\, V^\dagger\Lambda(|\psi\rangle\langle\psi|)V\rangle$ and similarly $\langle\sigma_y\rangle$ is its expectation value of $\sigma_y$. For the equalities, in $(i)$ we used the fact that any Pauli channel is unital and that $\max_\theta |1-e^{i\theta}|=2$, while  $(ii)$ is due to the relation $|0\rangle\langle1|=(\sigma_x+i\sigma_y)/2$.

As already mentioned, there is a  one-to-one correspondence between Pauli channels and Bell diagonal states, and a Pauli channel with corresponding Bell diagonal state described by $\{R_{ii}\}$ acts on a Bloch vector of components $n_i = Tr(\sigma_i \rho)/\sqrt{2}$, $i=x,y,z$ (see also Appendix \ref{sec: Bell and Pauli}), according to
\[
(n_x,n_y,n_z)\mapsto (R_{11}n_x,-R_{22}n_y,R_{33}n_z).
\]
Therefore its action on the Bloch sphere $S^2$ results in an ellipsoid of equatorial radii $(R_{11},R_{22},R_{33})$. 
Now rename $|R_{11}|,|R_{22}|,|R_{33}|$ according to their size as $\lambda_1, \lambda_2,\lambda_3$ with $1\geq \lambda_3\geq \lambda_2\geq \lambda_1\geq 0$.  Considering that $\sqrt{\langle\sigma_x\rangle^2+\langle\sigma_y\rangle^2}$ represents the Euclidean distance from the origin of the projection of the state $V^\dagger\Lambda(|\psi\rangle\langle\psi|)V$ on $xy$-plane, one can easily imagine that $Q_{\mathcal{N}}^{A}(\rho_{AB})$ is achieved when $V$ aligns the largest equatorial radius with the $z$-axis and choosing $\ket{\psi}$ to be the state mapped in the direction of the second largest equatorial radius (see Figure~\ref{fig:ellipsoid}). It is clear that this choice provides an upper bound  $Q_{\mathcal{N}}^{A}(\rho_{AB})\leq \lambda_2/2$. In the following we prove that this bound is saturated.\\

\begin{figure*}[th]
\includegraphics[width=14cm]{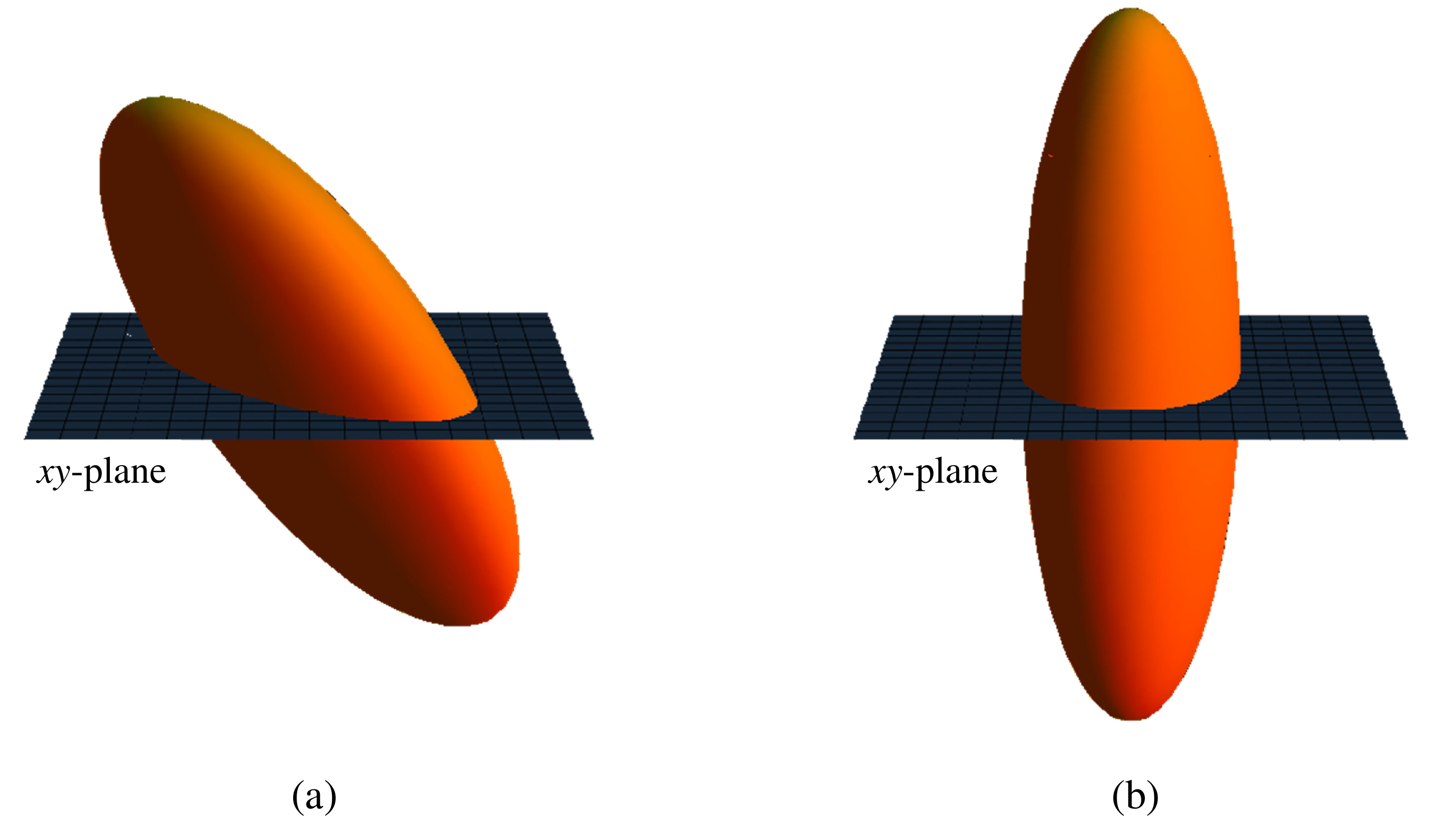}
\caption{(Color online) The problem of calculating the one-sided negativity of quantumness $Q^A_{\cal N}$ of Bell diagonal states cast in geometrical terms. (a) The surface of the ellipsoid is the set of states of the form $V^\dagger\Lambda(|\psi\rangle\langle\psi|)V$ with some unitary $V$ for the set of all pure qubit states $|\psi\rangle$. Here we want to find the optimal unitary $V$ such that it minimizes the maximum distance from the origin to its projection onto the $xy$-plane. (b) Clearly choosing the unitary which rotates the ellipsoid such that it aligns the longest equatorial radii along $z$-axis minimizes the maximum distance from the origin to its projection on $xy$-plane.}
\label{fig:ellipsoid}
\end{figure*}


First observe that an ellipsoid with equatorial radii $(\lambda_1,\lambda_2,\lambda_3)$ that are aligned with the coordinate axes can be expressed as
 \[
 \lbrace \underline{x}=(x_1,x_2,x_3)\microspace:\microspace\underline{x}^T A\underline{x}=1\rbrace,
 \]
 where $A=\diag{1/\lambda_1^2, 1/\lambda_2^2, 1/\lambda_3^2}$.
 Therefore an arbitrary ellipsoid with equatorial radii $(\lambda_1,\lambda_2,\lambda_3)$ can be expressed as
 \[
  \lbrace \underline{y}=(y_1,y_2,y_3):\underline{y}^T B\underline{y}=1\rbrace
  \]
with  $B=R^TAR$  for some rotation $R\in O(3)$. In particular, $B$ has the same---positive---eigenvalues as $A$. We then find
%
\begin{align*}
\max_{\underline{y}:\,\underline{y}^T B\underline{y}=1}&\sqrt{y_1^2 +y_2^2}\\
&\stackrel{(i)}{\geq}
\max\lbrace ||\underline{y}|| \microspace:\microspace\underline{y}=(y_1,y_2,0)\microspace,\microspace\microspace\underline{y}^TB\underline{y}=1\rbrace\\
&\stackrel{(ii)}{=} \max\lbrace (\underline{y}'^T B\underline{y}')^{-1/2}\microspace:\microspace \underline{y}'=(y'_1,y'_2,0)\microspace,\microspace\microspace\|\underline{y}'\|=1\}\\
&=1/\sqrt{ \min \lbrace (\underline{y}'^T B\underline{y}')\microspace:\microspace \underline{y}'=(y_1,y_2,0)\microspace,\microspace\microspace\|\underline{y}'\|=1\}}\\
&\stackrel{(iii)}{\geq}1/ \sqrt{1/\lambda_2^2}\\
&=\lambda_2,
\end{align*}
 where: $(i)$ holds because the cross-section of an arbitrary ellipsoid with the $xy$-plane is a subset of its projection onto the $xy$-plane; $(ii)$ holds because for any $\underline{y}$ satisfying $\underline{y}^TB\underline{y}=1$ we can consider a normalized $\underline{y}' = \underline{y} / \|\underline{y} \|$, and for any normalized $\underline{y}'$ we can consider $\underline{y}/\sqrt{ \underline{y}^TB\underline{y}}$ that satisfies by construction $\underline{y}^TB\underline{y}=1$; finally, $(iii)$ holds because Corollary  \MakeUppercase{\romannumeral 3}.1.2 of \cite{bhatia1997matrix} implies
 \begin{equation}\label{eq: bell-1}
 \frac{1}{\lambda_2^2}=\max_{\substack{\mathcal{M}\subset\real^3\\\dim{\mathcal{M}=2}}}\min_{\substack{\underline{x}\in\mathcal{M}\\||\underline{x}||=1}}\underline{x}^TB\underline{x},
 \end{equation}
 where $\mathcal{M}$ is a two-dimensional subspace of $\real^3$.

%

Therefore $Q_{\mathcal{N}}^{A}(\rho_{AB})=\frac{\lambda_2}{2}$.
\end{proof}


With this result for the one-sided NoQ we can also easily solve the two-sided case.
\begin{cor}[Two-sided NoQ for Bell diagonal states]

For a Bell diagonal state $\rho_{AB}$, its total NoQ is $Q_{\mathcal{N}}(\rho_A)=\frac{\lambda_2}{2}$
\end{cor}
\begin{proof}
Since local unitaries acts as $O(3)$ elements on the correlation matrix of $\rho_{AB}$, one can always transform it to $\begin{pmatrix}
1 & 0 & 0& 0\\
0 & \lambda_1 & 0 & 0\\
0 & 0 & \lambda_2 & 0\\
0 & 0 & 0 & \lambda_3 \\
\end{pmatrix}$ and a simple calculation shows that the $\|\cdot\|_{l_1}$-norm---in the computational basis---of a state with such correlation matrix is $1+\lambda_2$. Since $Q_{\mathcal{N}}^{A}(\rho_{AB})\leq Q_{\mathcal{N}}(\rho_{AB})$ \cite{piani:pr2012a}, this is the best we can achieve, i.e., $Q_{\mathcal{N}}(\rho_{AB})=\lambda_2/2$.
\end{proof}

\subsubsection{Extension of the analysis for Bell diagonal states}

For the analysis of Bell diagonal states, we made use of the fact that each Bell diagonal state can be expressed as a state generated by the action of a unique Pauli channel on the maximally entangled state. However there is no reason to restrict the channel to be a Pauli channel and NoQ can be computed in a similar manner for the state generated by the action of a general qubit channel (CPTP map) on the maximally entangled state.
\begin{thm}
Let $\rho_{AB}$ be a two-qubit state where $\rho_A=\tr_B[\rho_{AB}]$ is maximally mixed. Let $\bf{R}$ be the correlation matrix of the state $\rho_{AB}$ with elements $R_{\mu\nu}=\langle \sigma_\mu\otimes\sigma_\nu , \rho_{AB}\rangle$ and denote its $3\times 3-$submatrix as $\hat{\bf{R}}:=\lbrack R_{ij}\rbrack_{i,j\in\{1,2,3\}}$. Denote the singular values (including zeros) of $\hat{\bf{R}}$ in descending order as $s_1(\hat{\bf{R}})\geq s_2(\hat{\bf{R}}) \geq s_3(\hat{\bf{R}})$.
 Then $Q_{\mathcal{N}}^{A}(\rho_{AB})=s_2(\hat{\bf{R}})/2$.
\end{thm}

\begin{proof}
To apply the techniques from the analysis for the case of Bell diagonal states (Theorem~\ref{thm:onesidedbell}), it is important to notice the relation between the matrix $\bf{T}\in\real^{4\times 4}$ representing a Hermiticity-preserving linear map $\Omega$ (see Appendix~\ref{sec: qubit channels}) and the correlation matrix $\bf{R}$ of the corresponding operator $\rho_{AB}=(\I\otimes\Omega)[\phi_{AB}]$, with $\phi_{AB}$ the standard maximally entangled state. This relation can be expressed as
$R_{\mu\nu}=\frac{1}{2}(-1)^{\delta_{\mu2}}T_{\mu\nu}$, with $\delta_{\mu2}$ a Kronecker delta. Moreover local unitaries acting on $\rho_{AB}$ corresponds to the action of unitaries on the input and output  of the channel. Namely let $\bf{R}'$ be the correlation matrix of $(U_A\otimes U_B)\rho_{AB}(U_A\otimes U_B)^\dagger$ where $U_A$ and $U_B$ are local unitaries on system $A$ and $B$. Then the matrix $\bf{T}'$ defined  via $T'_{\mu\nu}=2 (-1)^{\delta_{\mu2}} R'_{\mu\nu}$ represents the action of $W_{U_B}\circ\Omega\circ W_{U_A^T}$, where $W_U$ stands for the conjugation by $U$, i.e., $W_U[X]=UXU^\dagger$. Also it is easy to see that the singular values of the $3\times 3$ submatrix $\hat{\bf{R}}=\lbrack R_{ij} \rbrack_{i,j\in\{1,2,3\}}$ and those of the submatrix of $\lbrack T_{ij}\rbrack_{i,j\in\{1,2,3\}}$ of $\bf{T}$ are equal up to the constant $1/2$.

Now we consider a state $\rho_{AB}$ maximally mixed on $A$, which can always---and uniquely---be represented as $\rho_{AB}=(\I_A\otimes\Lambda)[\phi_{AB}]$, with $\Lambda$ a channel. Now as explained in Appendix~\ref{sec: qubit channels}, one can find unitaries $U_A$ and $U_B$ such that $\Lambda=W_{U_B}\circ\tilde{\Lambda}\circ W_{U_A}$ with the matrix representation $\bf{T}$ of $\tilde{\Lambda}$ in the canonical form
\[
{\bf T}=\left(
\begin{BMAT}{c.ccc}{c.ccc}
1 & 0 & 0 & 0\\
t_1 & \lambda_1 && \\ t_2 &&\lambda_2 &\\ t_3&&&\lambda_3
\end{BMAT}
\right).
\]
Following the same steps taken in Eq.~\eqref{eq:bellfirststeps}, one finds
\[
\begin{split}
Q_{\mathcal{N}}^{A}(\rho_{AB})
&= \min_{\lbrace\ket{a_0},\ket{a_1}\rbrace} \big(\|\langle a_0|(\mathbb{I}_A\otimes\Lambda_B\lbrack\phi_{AB}\rbrack)|a_1\rangle\|_1\\ &\quad+\|\langle a_1|(\mathbb{I}_A\otimes\Lambda_B\lbrack\phi_{AB}\rbrack)|a_0\rangle\|_1\big)/2\\
&= \min_{V}\max_U |\bra{1}_BV^\dagger\Lambda^\dagger_B\lbrack U\rbrack V\ket{0}_B|/2
\end{split}
\]
where we have taken into account that now in general the channel $\lambda$ is not self-dual, so that $\Lambda^\dagger\neq \Lambda$.

From $\Lambda=W_{U_B}\circ\tilde{\Lambda}\circ W_{U_A}$, so that $\Lambda^\dagger = W_{U_A^\dagger} \circ\tilde{\Lambda}^\dagger \circ W_{U_B^\dagger}$, and following step $(i)$ of Eq. \eqref{eq:bellsecondsteps}, we arrive to
\begin{equation}
\label{eq:dualneg}
Q_{\mathcal{N}}^{A}(\rho_{AB})
=\min_{V}\max_{|\psi\rangle}  |\langle |0\rangle\langle1|,\, V^\dagger\tilde{\Lambda}^\dagger\lbrack|\psi\rangle\langle\psi|\rbrack V\rangle|,
\end{equation}
having used the fact that for any channel $\Lambda$, the dual map $\Lambda^\dagger$ is unital, i.e., $\tilde\Lambda^\dagger[\I]=\I$.

Now, the action of $\tilde{\Lambda}^\dagger$ on a state with Bloch coordinates $(1,w_1,w_2,w_3)$  is
\begin{equation}
\label{eq:dualbloch}
\tilde{\Lambda}^\dagger\left[\frac{1}{2}\left(\I+\sum_i w_i\sigma_i\right)\right]=\frac{1}{2}\left(\left(1+\sum_i t_iw_i\right)\I+\sum_i\lambda_iw_i\sigma_i\right),
\end{equation}
since the matrix representation of $\tilde{\Lambda}^\dagger$ is ${\bf T}^T$ (see Appendix~\ref{sec: qubit channels}).
While as soon as some $t_i$ is different from zero the map $\tilde\Lambda^\dagger$ is not trace-preserving, we see that this is irrelevant, as the first term on the right-hand side of Eq.~\eqref{eq:dualbloch} effectively does not contribute to the right-hand side of Eq.~\eqref{eq:dualneg}. Hence we can follow the proof of Theorem~\ref{thm:onesidedbell} as if we were dealing with a Pauli channel fully characterized by $\lambda_1,\lambda_2,\lambda_3$.
\end{proof}

\subsection{Werner states}

Here we present the formula for the one-sided and two-sided NoQ for general Werner states \cite{werner}.
\begin{defn}[Werner states]
Let $A$ and $B$ be $d$-dimensional quantum systems.
Then a Werner state  $\rho_{AB}$ is a bipartite state of the following form \cite{werner}
\[
 \rho_{AB}=\frac{\I_{AB}+\beta W}{d^2+d\beta},
 \]
where $\beta\in\real$ satisfies $|\beta|\leq 1$, $W=\sum_{ij}\state{i}{j}\otimes\state{j}{i}$ is the swap operator and $\I_{AB}$ is the identity operator on $AB$. Werner states are $UU$-invariant, i.e., $U_A\otimes U_B \rho_{AB}U^\dagger_A\otimes U^\dagger_B=\rho_{AB}$, for all $U$.
\end{defn}


\begin{thm}[NoQ of Werner states]
For any Werner state $\rho_{AB}=\frac{\I_{AB}+\beta W}{d^2+d\beta}$ the partial and the  total NoQ are both equal to
\[
Q_{\mathcal{N}}^{A}(\rho_{AB})=Q_{\mathcal{N}}^{AB}(\rho_{AB})=\frac{|\beta|(d-1)}{2(d+\beta)}.
\]
\end{thm}

\begin{proof}
First we prove the case for one-sided NoQ.

Recall from Eq.~\eqref{eq:onesidedNoQ}, $Q_\neg^A(\rho_{AB})=\min_{\bas{A}}\frac{1}{2}(\sum_{i,j}||\rho_{ij}||_1-1)$, where $\rho_{ij}=\rho_{ij}^B=\bra{a_i}\rho_{AB}\ket{a_j}$ for $\ket{a_i}=U\ket{i}$ elements of the basis $\bas{A}$, for some unitary $U$. It is clear that each $||\rho_{ij}||_1$ is invariant unitaries on $B$. We can then  use the $UU$-invariance of Werner states as follows:
\[
\begin{split}
\big\|\big\langle{a_i}\big\vert_A\rho_{AB}\big\vert{a_j}\big\rangle_A\big\|_1&= \|\bra{i}_AU^\dagger_A\rho_{AB}U_A\ket{j}_A\|_1 \\
&= \|\bra{i}_AU^\dagger_A(U_A\otimes U_B)\rho_{AB}(U_A\otimes U_B)^\dagger U_A\ket{j}_A\|_1 \\
&=  \|U_B\bra{i}_A\rho_{AB}\ket{j}_AU_B^\dagger \|_1\\
&= \|\bra{i}_A\rho_{AB}\ket{j}_A\|_1
\end{split}
\]
This proves that there is no need to perform the optimization to calculate the NoQ and one can use the computational basis on $A$. A straightforward calculation proves then the claim.

One can then calculate the two-sided NoQ checking that local measurements in the computational basis are optimal, since they allow to reach the lower bound constituted by the one-sided value.


\end{proof}


\subsection{Isotropic states}
Isotropic states  \cite{iso} are another class of bipartite states, amenable to an exact quantumness analysis.

\begin{defn}[Isotropic states]
Isotropic states are bipartite quantum state of the following form \cite{iso}:
\[
\rho_{AB}=\lambda\Phi+\frac{1-\lambda}{d^2-1}(\I-\Phi),
\]
where $d$ is the dimension of $A$ and $B$, $0\leq \lambda\leq 1$, and $\Phi=\state{\phi}{\phi}$ with $\ket{\phi}=1/\sqrt{d}\sum_{i=1}^d\ket{i}_A\ket{i}_B$ is the $d$-dimensional maximally entangled state. Isotropic states are $UU^*$-invariant, i.e., $U_A\otimes U^*_B \rho_{AB}U^\dagger_A\otimes U^T_B=\rho_{AB}$, for all $U$.
\end{defn}
\begin{thm}
The one-sided and two-sided NoQ of an isotropic state are both equal to
\[
Q_{\neg}^A(\rho_{AB})=Q_{\neg}^{AB}(\rho_{AB})=\frac{|\lambda d^2-1|}{d+1}.
\]
\end{thm}
\begin{proof}
The proof is essentially identical to the one for Werner states. One can use the $UU^*$-invariance to prove that the optimization in the calculation of the one-sided NoQ in unnecessary. Then a straightforward calculation considering the computational basis of $A$ leads to  $Q_{\neg}^A(\rho_{AB})=\frac{|\lambda d^2-1|}{d+1}$ (similarly for $B$). The result for the two-sided NoQ is obtained matching this lower bound by considering measurements in the two local computational bases.
%
%
\end{proof}

\section{Conclusions}\label{sec:concl}

In this paper we have quantitatively investigated a general notion of quantumness of correlations in bipartite and multipartite states \cite{modi:2011a}. Such quantumness can be related, for instance, to the disturbance induced on quantum states by local projective measurements.

We have reviewed several approaches to reveal and quantify the quantumness of correlations, proving that several of them are in fact equivalent in the general case of bipartite systems where local measurements act on a two-dimensional subsystem. We focused our analysis on a measure of quantumness of correlations defined as the minimum entanglement (measured by the negativity) created with a set of measurement apparatuses during the action of local measurements, following the so-called `activation' paradigm for nonclassical correlations \cite{piani:prl2011a,streltsov:prl2011b, doi:10.1142/S0219749911008258, piani:pr2012a}. The ensuing quantumness measure, known as negativity of quantumness, turns out to have very interesting properties. In particular, when the measured subsystem is a qubit it reduces to the minimum disturbance as measured by the trace distance or, alternatively, to the minimum trace distance to states that are classical on the qubit. We clarified the mechanism of the activation protocol for negativity and proved a bound on the negativity of arbitrary bipartite states conjectured in \cite{khasin:pr2007a}.

We finally presented a number of examples on which the negativity of quantumness can be computed exactly. These include relevant families of states such as Werner, isotropic, and two-qubits states that have one maximally-mixed marginal. The latter class not only includes all Bell diagonal states, but also all the states isomorphic to an arbitrary qubit channel $\Lambda$ via $(\openone_A\otimes\Lambda_B)[\phi_{AB}]$, with $\phi_{AB}$ the standard maximally entangled state of two qubits. Given the hierarchical relation
\[
Q_{\neg}^A(\rho_{AB})\geq \neg (\rho_{AB})
\]
of \cite{piani:pr2012a}, the closed formula of this paper allows, e.g., a consistent study and comparison of the evolution of entanglement---as measured by negativity---and quantumness of correlations---as measured by the one-sided negativity of quantumness---under the action of a family of qubit channels, e.g., a semigroup. We remark that, while $Q_{\neg}^A$ could increase under actions on $A$, both $Q_{\neg}^A$ and $\neg$ can only decrease under actions on $B$.

We believe the unveiled connections between apparently unrelated approaches to define and quantify general nonclassical correlations might inspire further research into the rationale of quantum measurements, possibly bringing to a better understanding of the most essential features which mark a departure from a classical description of nature. From a practical perspective, the negativity of quantumness has been already linked to the performance of the remote state preparation primitive in noisy one-way quantum computations \cite{chavesito}. We can expect more general frameworks to be defined in the near future where the quantumness of correlations, perhaps measured by the negativity of quantumness, can emerge as resource to beat classical strategies for some relevant task. Quantum communication and metrology seem fertile grounds for such an expectation to grow into practice. We finally remark that the negativity of quantumness can be bounded by experimentally accessible witnesses \cite{vianna}. An experimental demonstration of the activation of nonclassical correlations into entanglement during local measurements is under way \cite{fabio}.

\acknowledgements{
We are grateful to  D. Girolami, T. Tufarelli, V. Giovannetti, F. Illuminati, R. Lo Franco, G. Compagno, and R. Vianna for helpful discussions. We thank J. Watrous for providing the main argument of Lemma~\ref{lem:John}.
This work was supported by  CIFAR, NSERC, the Ontario Centres of Excellence and the UK EPSRC. We acknowledge in particular joint support by the University of Nottingham Research Development Fund (Pump Priming grant 0312/09).
GA and TN thank respectively the Institute of Quantum Computing and the University of Nottingham for the kind hospitality during mutual visits in which most of this work was developed.
}

\appendix

\section{Properties of the \texorpdfstring{$l_1$}{}-norm}\label{sec: l1-norm}

The $l_1$-norm (sometimes called taxicab metric) of a vector is defined as the sum of absolute values of all entries. For $\underline{x}=(x_1,x_2,\cdots, x_n)\in \complex^n$,
\[
\lnorm{\underline{x}}=\sum_{i=1}^n|x_i|.
\]
We will be interested in applying this kind of norm to $n\times n$ matrices with complex entries $A\in\matsize{n}$, so that
\beq
\label{eq:l1norm}
\lnorm{A}=\sum_{i,j=1}^n|A_{ij}|.
\eeq
One advantage of the norm~\eqref{eq:l1norm} is its ease of calculation, but it has some drawbacks. First, it is not sub-multiplicative, i.e., it does not respect $\lnorm{AB} \leq \lnorm{A}\lnorm{B}$. Moreover this norm is not invariant under conjugation by unitaries, which implies a matrix takes a different $l_1$-norm value depending on the basis chosen for its representation. To take this into account, the basis with respect to which the $l_1$-norm is calculated is indicated as superscript when needed.

One notable feature of the $l_1$-norm which is used in the paper is the following.
\begin{lem}
\label{lem:John}
For any $A\in\matsize{n}$ it holds
\begin{equation}
\label{eq:lineq}
\lnorm{A}^\mathcal{B}\geq ||A||_{1}
\end{equation}
independently of the basis $\mathcal{B}$ in which the $l_1$-norm is calculated.

If $A$ is normal, i.e. $AA^\dagger=A^\dagger A$, then is a choice of basis $\bar{\mathcal{B}}$ such that the minimum in \eqref{eq:lineq} is achieved, i.e.
\begin{equation}
\lnorm{A}^{\bar{\mathcal{B}}}= ||A||_{1}\label{eq:leq}
\end{equation}

\end{lem}
\begin{proof}
For a fixed choice of basis $\mathcal{B}$---the one in which we calculate $\lnorm{\cdot}^\mathcal{B}$---consider the subset of matrices $\Omega^\mathcal{B}$ whose entries have modulus equal or smaller than 1, i.e., $\Omega^\mathcal{B}=\{[b_{ij}]_{i,j\in\upto{n}}\in\matsize{n}\> ;\> |b_{ij}|\leq 1\>,\> \forall i,j\in\upto{n}\}$. The $l_1$-norm of $A$ in the basis $\mathcal{B}$ can be written as
\begin{equation}
\|A\|^\mathcal{B}_{1}=\max_{V \in \Omega} |\Tr(VA)|.
\end{equation}
On the other hand, the 1-norm can be written as
\begin{equation}
||A||_{1}=\max_{V\in\mathbb{U}(n)}|\tr(VA)|.
\end{equation}
Observe that $\mathbb{U}(n)\subset\Sigma$, because the rows and columns of a unitary matrix form a set of orthonormal vectors. Therefore Eq.\eqref{eq:lineq} holds.

Eq.\eqref{eq:leq} is trivial because a normal matrix can be diagonalized by a change of basis.
\end{proof}



\section{The equality of entanglement and quantumness for MCS}\label{sec:equality of ent and quant}

\begin{thm}
Let $\rho_{AB}$ be a MCS. Then
\[
Q_E^{AB}(\rho_{AB})=Q_E^{A}(\rho_{AB})=E_{A:B}(\rho_{AB})
\]
with $Q_E^{AB}$ and $Q_E^{A}$ the measures of quantumness (two-sided and one-sided, respectively) defined by means of~\eqref{eq:defquantent} and $E_{A:B}$ the entanglement between $A$ and $B$.
\end{thm}
\begin{proof}
By the result of~\cite{piani:pr2012a}, $Q_E^{AB}(\rho_{AB})\geq Q_E^{A}(\rho_{AB})\geq E_{A|B}(\rho_{AB})$ holds. Therefore it is sufficient to prove $E_{A|B}(\rho_{AB})\geq Q_E^{AB}(\rho_{AB})$. By definition $Q_E^{AB}(\rho_{AB}):=\min E_{AB|A'B'}(\tilde{\rho}_{ABA'B'})$ where the minimum is taken over the choice of different pre-measurement state. Now we can choose a particular measurement interaction which acts on the basis of the maximally correlated form, i.e.,
\[
\sum_{ij}\rho_{ij}\state{a_i}{a_j}\otimes \state{b_i}{b_j}\mapsto \sum_{ij}\rho_{ij}\state{a_i}{a_j}\otimes \state{b_i}{b_j}\otimes\state{i}{j}_{A'}\otimes \state{i}{j}_{B'}.
\]
The resultant state is equivalent to the original $\rho_{AB}$ up to local unitary on $AB$ or $A'B'$. Therefore $E_{A|B}(\rho_{AB})\geq Q_E^{AB}(\rho_{AB})$.
\end{proof}

\section{Proof of Theorem~\ref{thm: closestclassicalMCS}}\label{sec: proof}
\begin{proof}
We simply denote $\tau_{AB}$ and $\sigma_{AB}$ as $\tau,\sigma$.
Let  $\eta=\sum_{klmn}\lambda_{klmn}|a_k\rangle\langle a_l|\otimes |b_m\rangle\langle b_n|$ be a traceless Hermitian matrix such that $\xi=\tau+\eta$ is an arbitrary separable state. We want to prove that
\[
||\tau-\sigma||^{\{\ket{a_i}\otimes\ket{b_j}\}}_{l_1}\leq \|\tau-\xi\|^{\{\ket{a_i}\otimes\ket{b_j}\}}_{l_1}=\|\eta\|^{\{\ket{a_i}\otimes\ket{b_j}\}}_{l_1}.
\]
Since $||\tau-\sigma||^{\{\ket{a_i}\otimes\ket{b_j}\}}_{l_1}=\sum_{i\neq j}|\tau_{ij}|$, it is sufficient to prove that $\sum_{i\neq j}|\tau_{ij}|\leq \sum_{klmn}|\lambda_{klmn}|$.

In our argument we will make use of the negative eigenvectors of $\tau^\Gamma$
\[
\ket{\phi_{ij}}=\frac{1}{\sqrt{2}}(\ket{a_i}\otimes \ket{b_j}- \frac{\tau_{ji}}{|\tau_{ji}|}(\ket{a_j}\otimes \ket{b_i})),
\]
where $i\neq j$, with corresponding eigenvalues $-|\tau_{ij}|$ (see Lemma~\ref{lem:spectrumMCS}). 

We now consider two cases:
\begin{enumerate}
\item
Suppose the diagonal entries of $\xi$ are the same as those of $\tau$, i.e.,
\[
\lambda_{ijkl}=0\quad\text{for}\quad i=j,k=l.
\]
Now consider the partial transpose of $\xi$, $\xi^\Gamma = \tau^\Gamma+\eta^\Gamma= \tau^\Gamma+\sum_{klmn}\lambda_{klmn}|a_k\rangle\langle a_l|\otimes |b_n\rangle\langle b_m|$. It holds	
\begin{align*}
\langle \phi_{ij}|\xi^\Gamma|\phi_{ij}\rangle &= \langle \phi_{ij}|\tau^\Gamma|\phi_{ij}\rangle+\langle \phi_{ij}|\sum_{klmn}\lambda_{klmn}|a_k\rangle\langle a_l|\otimes |b_n\rangle\langle b_m|\phi_{ij}\rangle
\\
&= -|\tau_{ij}|-\frac{1}{2|\tau_{ij}|}(\comp{\tau_{ji}}\comp{\lambda_{ijij}}+\tau_{ji}\lambda_{ijij})
\\
&= -|\tau_{ij}|-\Re\left[\frac{\tau_{ji}}{|\tau_{ij}|}\lambda_{ijij}\right].
\end{align*}
By assumption $\xi$ is a separable state, i.e., the numerical range of $\xi^\Gamma$ is in positive real line. Therefore a necessary condition for $\xi$ to be PPT is $\langle \phi_{ij}|\xi^\Gamma|\phi_{ij}\rangle\geq 0$ for all $i\neq j$.\\
We find
\begin{align*}
\sum_{i\neq j}|\tau_{ij}| &\leq \sum_{i\neq j}-\Re\left[\frac{\tau_{ji}}{|\tau_{ij}|}\lambda_{ijij}\right]
\\
& \leq \sum_{i\neq j}|\lambda_{ijij}|
\\
& \leq \sum_{klmn}|\lambda_{klmn}|,
\end{align*}
as claimed.

\item Consider an arbitrary $\eta$, i.e., no conditions are imposed on the coefficients $\lambda_{klmn}$ except that they lead to a separable state $\xi$. Then there are more terms in the expression for $\langle \phi_{ij}|\xi^\Gamma|\phi_{ij}\rangle$ than those encountered in the previous calculation. Namely,
\begin{align*}
\langle \phi_{ij}|\xi^\Gamma|\phi_{ij}\rangle &= -|\tau_{ij}|-\frac{1}{2|\tau_{ij}|}(\comp{\tau_{ji}}\comp{\lambda_{ijij}}+\tau_{ji}\lambda_{jjii})+\frac{1}{2}(\lambda_{iijj}+\comp{\lambda_{jjii}})
\\
&= -|\tau_{ij}|-\Re[\frac{\tau_{ji}}{|\tau_{ij}|}\lambda_{ijij}]+\frac{1}{2}(\lambda_{iijj}+\comp{\lambda_{jjii}})
\end{align*}
Nonetheless, by imposing $\langle \phi_{ij}|\xi^\Gamma|\phi_{ij}\rangle\geq 0$ for all $i\neq j$ we find:
\begin{align*}
\sum_{i\neq j}|\tau_{ij}| &\leq \sum_{i\neq j}(-\Re[\frac{\tau_{ji}}{|\tau_{ij}|}\lambda_{ijij}]+\frac{1}{2}(\lambda_{iijj}+\comp{\lambda_{jjii}})
\\
&\leq\sum_{i\neq j}(|\lambda_{ijij}|+|\lambda_{iijj}|)
\\
&\leq \sum_{klmn}|\lambda_{klmn}|
\end{align*}

\end{enumerate}
Therefore for both cases, $\sum_{i\neq j}|\tau_{ij}|$ is the smallest possible value and $\sigma$ is one of the closest separable state.

\end{proof}

\section{The isometric mapping for quantum relative entropy}\label{sec:isometric for rel}

Relative entropy of entanglement is an entanglement measure defined as the ``distance''---in term of quantum relative entropy (see Section~\ref{sec:notation})---from the closest separable state \cite{vedral:prl1997a}:
\begin{equation}
E_R(\rho_{AB})=\min_{\sigma_{AB}\in\sep}\rel{\rho_{AB}}{\sigma_{AB}}.
\end{equation}

The measure of nonclassical correlation based on quantum relative entropy is called relative entropy of discord \cite{modi:prl2010a} and is defined as
\begin{equation}
E_R(\rho_{AB})=\min_{\sigma_{AB}\in\CC}\rel{\rho_{AB}}{\sigma_{AB}},
\end{equation}
and in \cite{modi:prl2010a} its equivalence to zero-way quantum deficit~\eqref{eq:zeroway} was proved. Namely,
\begin{equation}
\min_{\sigma_{AB}\in\CC}\rel{\rho_{AB}}{\sigma_{AB}}=\min_{\Pi_A\otimes\Pi_B}\rel{\rho_{AB}}{\Pi_A\otimes\Pi_B [\rho_{AB}]}
\end{equation}
Since quantum relative entropy is invariant under linear isometry,
\begin{equation}
\min_{\Pi_A\otimes\Pi_B}\rel{\rho_{AB}}{\Pi_A\otimes\Pi_B [\rho_{AB}]}=\min_{\Pi_A\otimes\Pi_B}\rel{\tilde{\rho}_{ABA'B'}}{\Pi_A\otimes\Pi_B [\tilde{\rho}_{ABA'B'}]},
\end{equation}
where $\tilde{\rho}_{ABA'B'}$ is the pre-measurement state constructed by the measurement interaction on the basis $\Pi_A\otimes \Pi_B$.
Finally by the exact expression of relative entropy of entanglement for MCS given in \cite{rains:itit2001a}, one can deduce that $\Pi_A\otimes\Pi_B [\tilde{\rho}_{ABA'B'}]$ is indeed one of the closest separable state for $\rho_{ABA'B'}$.

Eq.~\eqref{diagrammadue} summarizes the relations and equivalences just explained.
\begin{equation}
\label{diagrammadue}
\begin{diagram}[height=2em,width=0.3em]
$${
\displaystyle{\min_{{\Pi_{A}\otimes\Pi_{B}}}}
\rel{\tilde{\rho}_{{ABA'B'}}}{(\Pi_A\otimes\Pi_B)[\tilde{\rho}_{{ABA'B'}}]}}
$$ &=^{(\romannumeral 2)} & $${\displaystyle{\min_{{\eta\in\sep}}}
\rel{\tilde{\rho}_{ABA'B'}}{\eta}}$$\\
=_{(\romannumeral 1)} & &\\
$${
\displaystyle{\min_{{\Pi_A\otimes\Pi_B}}}\rel{\rho_{AB}}{(\Pi_A\otimes\Pi_B)[\rho_{AB}]}}$$ & =_{(\romannumeral 3)}&  $$\>{\displaystyle{\min_{\sigma\in\CC}}\rel{\rho_{AB}}{\sigma}}$$\\
&&
\end{diagram}
\end{equation}
In Eq.~\eqref{diagrammadue}, (\rmnum{1}) is because the measurement interaction for the basis $\Pi_A\otimes\Pi_B$ is isometric, (\rmnum{2}) is because the closest separable state of MCS is its diagonal part, (\rmnum{3}) is because the closest $\CC$ state is again its diagonal part.

\section{Bell diagonal states and Pauli channels}\label{sec: Bell and Pauli}

Since the set of pure qubit states corresponds to the complex projective space $\complex \mathbf{P}^1$ and there is an isomorphism between the unit sphere $\mathbf{S}^2\subset\real^3$ and $\mathbb{C} \mathbf{P}^1$, the states of a qubit can be represented as points in a unit ball $\mathbf{B}^2$. Namely for a qubit state with a density matrix  $\rho$ define a vector $\underline{n}=(n_1,n_2,n_3)\in\real^3$ as $n_i=\tr[\sigma_i\rho]$ where $\sigma_1=\sigma_x,\sigma_2=\sigma_y,\sigma_3=\sigma_z$ are the Pauli matrices (see also Section~\ref{sec:notation}). The vector $\underline{n}$ is called the Bloch vector of the qubit state $\rho$. A pure $\rho$ corresponds to a unit vector $\underline{n}$, while a mixed state have $||\underline{n}||<1$. That is, pure states corresponds to $\mathbf{S}^2$ and mixed states to its interior.
Conversely one can recover the density matrix associated to a Bloch vector $\underline{n}$ via
\[
\rho=\frac{1}{2}(\I+\underline{n}\cdot\underline{\sigma}),
\]
where $\underline{n}\cdot\underline{\sigma}=\sum_{i=1}^3n_i\sigma_i$.

This representation allows us to geometrically analyze qubit states. For example, when a Pauli channel of the form $\Lambda[\rho]=\sum_{\mu=0}^{3}p_i\sigma_\mu\rho\sigma_\mu$ with $\{p_i\}$ a probability vector and $\sigma_0=\I$ acts on a qubit sate $\rho$, the Bloch vector $\underline{n}$ of $\rho$ transforms as
\begin{eqnarray*}
(n_1,n_2,n_3)&\mapsto&\big((p_0+p_1-p_2-p_3)n_1,\\
&&\ (p_0-p_1+p_2-p_3)n_2,\\
&&\ (p_0-p_1-p_2+p_3)n_3\big).
\end{eqnarray*}


 Bell diagonal states have some notable properties.
First, a Bell diagonal state can be expressed as the action of a Pauli channel on a Bell state.  More precisely, there is a one-to-one relation between the set of Pauli channels and the set of Bell diagonal states:
\begin{eqnarray*}
&&p_0\ket{\phi^+}\bra{\phi^+}+p_1\ket{\psi^+}\bra{\psi^+}+p_2\ket{\psi^-}\bra{\psi^-}+p_3\ket{\phi^-}\bra{\phi^-}\\
&=&(\Lambda\otimes\I)[\ket{\phi^+}\bra{\phi^+}],
\end{eqnarray*}
where $\I$ indicates here the identity channel on a qubit and $\Lambda[\rho]=\sum_{\mu=0}^{3}p_i\sigma_\mu\rho\sigma_\mu$ as defined before.
 Second, the correlation matrix of Bell diagonal states have diagonal form.
A simple algebra shows that the correlation matrix of Bell diagonal states satisfies the following relations:
\begin{align*}
&R_{00} = 1\,,\,R_{ij}=0 \quad\text{for}\quad i\neq j,
\\
&R_{11}+R_{22}+R_{33} \leq 1,
\\
&1+R_{11}+R_{22} \geq R_{33},
\\
&1+R_{11}+R_{33} \geq R_{22},
\\
&1+R_{22}+R_{33} \geq R_{11}.
\end{align*}
Indeed the restrictions above forces the vector $(R_{11},R_{22},R_{33})$ to be within a tetrahedron in $\real^3$ \cite{horodecki1996information}.
The two different pictures of Bell diagonal states, the Pauli channel representation and the correlation matrix representation,  are related in the following way:
\[
\left( \begin{array}{c} R_{00} \\ R_{11} \\ R_{22} \\ R_{33} \end{array} \right)
=\left( \begin{array}{cccc} 1 & 1 & 1 & 1 \\1 & 1 & -1 & -1 \\ -1 & 1 & -1 & 1\\ 1 & -1 & -1 & 1 \end{array} \right).
\left( \begin{array}{c} p_0 \\ p_1 \\ p_2 \\ p_3 \end{array} \right)
\]


\section{Canonical form of qubit channels}\label{sec: qubit channels}

\subsection{Matrix representation of qubit channels}
The Bloch parametrization $n_\mu=\Tr(\sigma_\mu X)$, $\mu=0,\ldots,3$, of a Hermitian matrix $X=X^\dagger\in\matsize{2}$ allows us to represent such Hermitian matrices---and in particular states, for which $n_0=1$---as vectors in $\real^4$. Any Hermiticity-preserving map $\Omega:\complex^{2\times2}\rightarrow\complex^{2\times2}$ can then be represented as matrix $\bf{T}\in\real^{4\times 4}$, with $T_{ij}=\langle \sigma_i,\Omega\lbrack\sigma_j\rbrack\rangle$. If $\Omega$ is a channel, so that in particular it preserves trace, then $T_{01}=T_{02}=T_{03}=0$. Also, it is easy to check that the matrix representing the dual of $\Omega$ (see Section~\ref{sec:introduction}) is given by the transpose ${\bf{T}}^T$ of the original matrix representation $\bf{T}$ of $\Omega$.

\subsection{The canonical form of qubit channels~\cite{beth2002analysis}}
Let $\Gamma$ be an arbitrary Hermiticity- and trace-preserving linear map on qubits represented by $\bf{T}\in\real^{4\times 4}$. Then one can find suitable local unitaries $U_A$ and $U_B$ such that the ${\bf{T}}$-representation of the new channel $\Gamma'=W_{{U}_{A}}\circ\Gamma\circ W_{U_B}$, with $W_{U}$ acting by conjugation, i.e. $W_{U}[X]=U X U^\dagger$, has the canonical form
\[
\left(
\begin{BMAT}{c.ccc}{c.ccc}
1 & 0 & 0 & 0\\
t_1 & \lambda_1 && \\ t_2 &&\lambda_2 &\\ t_3&&&\lambda_3
\end{BMAT}
\right),
\]
where the $\lambda_i$'s are the singular values of the $3\times3$ real submatrix $[T_{ij}]_{ij}$, $i,j=1,2,3$.
The conjugate channel of $\Gamma'$ is $ \Gamma'^\dagger=W_{U^\dagger_B}\circ\Gamma^\dagger\circ W_{U^\dagger_A}$ and is represented by the matrix
\[
\left(
\begin{BMAT}{c.ccc}{c.ccc}
1 & t_1 & t_2 & t_3\\
0 & \lambda_1 && \\ 0 &&\lambda_2 &\\ 0&&&\lambda_3
\end{BMAT}
\right).
\]

\bigskip


\end{document}